\documentclass{article}
\usepackage{framed}
\usepackage{amssymb,amsfonts,amsmath}
\usepackage{algorithm2e,framed}
\usepackage{amsfonts,amsmath}
\usepackage{graphicx}
\usepackage{subfigure,enumerate,bm,amsmath,amsthm,wrapfig,array,color,algorithmic}
\usepackage{graphics}
\usepackage{multirow}
\usepackage{setspace}
\usepackage{epstopdf}

\usepackage{palatino}
\usepackage{fullpage}

\long\def\symbolfootnote[#1]#2{\begingroup%
\def\thefootnote{\fnsymbol{footnote}}\footnote[#1]{#2}\endgroup}

\DeclareMathOperator{\Expect}{\mathbb{E}}
\newcommand{\Trace }[1]{\mbox{}{\bf{Tr}}\left(#1\right)}

\newcommand{\FNorm }[1]{\mbox{}\|#1\|_{\mathrm{F}}  }
\newcommand{\FNormS}[1]{\mbox{}\|#1\|_{\mathrm{F}}^2}

\newcommand{\TNorm }[1]{\mbox{}\|#1\|_2  }
\newcommand{\TNormS}[1]{\mbox{}\|#1\|_2^2}

\newtheorem{theorem}{\bf Theorem}[]
\newtheorem{lemma}[theorem]{Lemma}
\newtheorem{definition}[theorem]{Definition}

\newcommand{\transp}{^{\textsc{T}}}

\newcommand{\mat}[1]{ {\ensuremath{\mathbf{#1} }}}

\newcommand{\abs }[1]{\left|#1\right|}

\newcommand{\var}[1]{\text{Var}\ensuremath{\left[#1\right] } }

\def\rank{\hbox{\rm rank}}

\def\p{{\mathbf p}}
\def\e{{\mathbf e}}

\def\matA{\mat{A}}
\def\matB{\mat{B}}
\def\matC{\mat{C}}
\def\matD{\mat{D}}
\def\matE{\mat{E}}
\def\matH{\mat{H}}
\def\matI{\mat{I}}
\def\matP{\mat{P}}
\def\matQ{\mat{Q}}
\def\matR{\mat{R}}
\def\matS{\mat{S}}

\def\matU{\mat{U}}
\def\matV{\mat{V}}
\def\matW{\mat{W}}
\def\matX{\mat{X}}
\def\matY{\mat{Y}}
\def\matZ{\mat{Z}}
\def\matSig{\mat{\Sigma}}
\def\matOmega{\mat{\Omega}}

\def\matTheta{\mat{\Theta}}

\DeclareMathSymbol{\Prob}{\mathbin}{AMSb}{"50}
\newcommand\remove[1]{}

\def\math#1{$#1$}

\def\frac#1#2{{#1\over #2}}

\def\eqan#1{\begin{eqnarray*}
#1
\end{eqnarray*}}

\DeclareMathSymbol{\R}{\mathbin}{AMSb}{"52}

\def\cl#1{{\cal #1}}

\def\argmin{\mathop{\hbox{argmin}}\limits}

\def\x{{\mathbf x}}

\def\a{{\mathbf a}}

\def\norm#1{{\|#1\|}}

\def\ceil#1{{\left\lceil\,#1\,\right\rceil}}

\def\dotfil{\leaders\hbox to 1.5mm{.}\hfill}
\newcounter{rmnum}
\def\RN#1{\setcounter{rmnum}{#1}\uppercase\expandafter{\romannumeral\value{rmnum}}}
\def\rn#1{\setcounter{rmnum}{#1}\expandafter{\romannumeral\value{rmnum}}}

\newcommand{\vct}[1]{\bm{#1}}

\newcommand{\ignore}[1]{}

\usepackage{float}

\usepackage{hyperref}
\hypersetup{
    unicode=false,          
    pdftoolbar=true,        
    pdfmenubar=true,        
    pdffitwindow=true,      
    pdftitle={},    
    pdfauthor={},     
    pdfsubject={},   
    pdfcreator={},   
    pdfproducer={}, 
    pdfkeywords={}, 
    pdfnewwindow=true,      
    colorlinks=false,       
    linkcolor=red,          
    citecolor=green,        
    filecolor=magenta,      
    urlcolor=cyan          
}

\begin{document}

\title{\bf Randomized Dimensionality Reduction for $k$-means Clustering}

\author{
{\bf Christos Boutsidis} \\
Yahoo Labs \\
New York, NY \\
\texttt{boutsidis@yahoo-inc.com}
\and
{\bf Anastasios Zouzias}\\
Mathematical \& Computational Sciences\\
IBM Z\"{u}rich Research Lab\\
\texttt{azo@zurich.ibm.com}
\and
{\bf Michael W. Mahoney}\\
Department of Statistics\\
UC Berkeley\\
\texttt{mmahoney@stat.berkeley.edu}
\and
{\bf Petros Drineas}\\
Computer Science Department \\
Rensselaer Polytechnic Institute \\
\texttt{drinep@cs.rpi.edu}
}

\maketitle

\begin{abstract}
We study the topic of dimensionality reduction for $k$-means clustering. Dimensionality reduction encompasses the union of two approaches: \emph{feature selection} and \emph{feature extraction}. A feature selection based algorithm for $k$-means clustering selects a small subset of the input features and then applies $k$-means clustering on the selected features. A feature extraction based algorithm for $k$-means clustering constructs a small set of new artificial features and then applies $k$-means clustering on the constructed features. Despite the significance of $k$-means clustering as well as the wealth of heuristic methods addressing it, provably accurate feature selection methods for $k$-means clustering are not known. On the other hand, two provably accurate feature extraction methods for $k$-means clustering are known in the literature; one is based on random projections and the other is based on the singular value decomposition (SVD).
This paper makes further progress towards a better understanding of dimensionality reduction for $k$-means clustering. Namely, we present the first provably accurate feature selection method for $k$-means clustering and, in addition, we present two feature extraction methods. The first feature extraction method is based on random projections and it improves upon the existing results in terms of time complexity and number of features needed to be extracted. The second feature extraction method is based on fast approximate SVD factorizations and it also improves upon the existing results in terms of time complexity. The proposed algorithms are randomized and provide constant-factor approximation guarantees with respect to the optimal $k$-means objective value.
\end{abstract}

\section{Introduction}\label{sec:intro}
Clustering is ubiquitous in science and engineering with numerous application domains ranging from bio-informatics and
medicine to the social sciences and the web~\cite{Har75}. Perhaps
the most well-known clustering algorithm is the so-called
``$k$-means'' algorithm or Lloyd's method \cite{Llo82}. Lloyd's method is an
iterative expectation-maximization type approach that attempts
to address the following objective: given a set of Euclidean points and a positive integer $k$ corresponding to the number of
clusters, split the points into $k$ clusters so that the total
sum of the squared Euclidean distances of each point to its
nearest cluster center is minimized. Due to this intuitive objective as well as its \emph{effectiveness}~\cite{ORSS06},  the 
Lloyd's method for $k$-means clustering has become enormously popular in applications~\cite{Wu07}.  

In recent years, the high dimensionality of modern massive
datasets has provided a considerable challenge to the design of efficient algorithmic solutions for $k$-means clustering. First, ultra-high dimensional data force existing algorithms for $k$-means clustering to be computationally inefficient, and second, the
existence of many irrelevant features may not allow the
identification of the relevant underlying structure in the data
\cite{GGBD05}. Practitioners have addressed these obstacles by
introducing feature selection and feature extraction techniques.
Feature selection selects a (small) subset
of the actual features of the data, whereas feature
extraction constructs a (small) set of artificial features based on the original features. Here, we consider a rigorous approach to feature selection and feature extraction for $k$-means clustering. Next, we describe the mathematical framework under which we will study such dimensionality reduction methods.

Consider $m$ points $\mathcal{P} = \{ \p_1, \p_2, \dots, \p_m \} \subseteq \R^n$ and an integer
$k$ denoting the number of clusters. The objective of
$k$-means is to find a $k$-partition of
$\mathcal{P}$ such that points that are ``close'' to each other belong to the same cluster and points
that are ``far'' from each other belong to different clusters.
A $k$-partition of $\mathcal{P}$ is a collection
$\cl{S}=\{\mathcal{S}_1, \mathcal{S}_2, \dots, \mathcal{S}_k\}$
of \math{k} non-empty pairwise disjoint
sets which covers \math{\cl{P}}.
Let $s_j=|\mathcal{S}_j|$ be the size of $\mathcal{S}_j$ ($j=1,2,\dots,k$).
For each set \math{S_j}, let \math{\bm\mu_j\in\R^n} be its centroid:
$$\bm\mu_j= \frac{\sum_{\p_i\in S_j}\p_i}{s_j}.$$
The $k$-means objective function is
$$
\mathcal{F}(\mathcal{P}, \cl{S}) =
\sum_{i=1}^m\norm{\p_i-\bm\mu(\p_i)}_2^2,
$$
where \math{\bm\mu(\p_i) \in \R^n} is the centroid of the cluster to which \math{\p_i} belongs.
The objective of $k$-means clustering is to compute the optimal $k$-partition of the points in $\mathcal{P}$,
$$\cl{S}_{opt}= \argmin_{\cl{S}} \cl{F} (\mathcal{P}, \cl{S}).$$
Now, the goal of dimensionality reduction for $k$-means clustering is to construct points
$$\mathcal{\hat{P}} = \{ \hat{\p}_1, \hat{\p}_2, \dots, \hat{\p}_m \} \subseteq \R^r$$
(for some parameter $r \ll n$) so that $\mathcal{\hat{P}}$ approximates the clustering structure of $\mathcal{P}$. Dimensionality reduction via feature selection constructs the $\hat{\p}_i$'s by selecting actual features of
the corresponding $\p_i$'s, whereas dimensionality reduction via feature extraction constructs new artificial features based on the original features. More formally, assume that the optimum $k$-means partition of the points in $\mathcal{\hat{P}}$ has been computed
\[
\hat{\cl{S}}_{opt} = \argmin_{\cl{S}} \cl{F} (\mathcal{\hat{P}}, \cl{S}).
\]
A dimensionality reduction algorithm for $k$-means clustering constructs a new set $\mathcal{\hat{P}}$ such that
\[ \cl{F} (\mathcal{P}, \hat{\cl{S}}_{opt}) \leq \gamma \cdot \cl{F} (\mathcal{P}, \cl{S}_{opt})\]
where $\gamma>0$ is the approximation ratio of $\hat{\cl{S}}_{opt}$. In other words, we require that computing an optimal partition $\hat{\cl{S}}_{opt}$ on the projected low-dimensional data and plugging it back to cluster the high dimensional data, will imply a $\gamma$ factor approximation to the optimal clustering. Notice that we measure approximability by evaluating the $k$-means objective function, which is a well studied approach in the literature~\cite{OR99,KSS04,HM04,PK05,FS06,ORSS06,AV07}. Comparing the structure of the actual clusterings $\hat{\cl{S}}_{opt}$ to
$\cl{S}_{opt}$ would be much more interesting but our techniques do not seem to be helpful towards this direction.
However, from an empirical point of view~(see Section~\ref{sec:experiments}), we do compare $\hat{\cl{S}}_{opt}$ directly to $\cl{S}_{opt}$ observing favorable results. 

\begin{table}[htdp]
\begin{center}
\begin{tabular}{|c|c|c|c|c|}
\hline
\textbf{Reference}    & \textbf{Description}     & \textbf{Dimensions} & \textbf{Time = $O(x), x =$} & \textbf{Approximation Ratio}          \\
\hline
 Folklore        & RP  & $O(\log(m)/\varepsilon^2)$        & $mn \lceil \varepsilon^{-2} \log(m)/ \log(n) \rceil$   & $1+\varepsilon$ \\
\hline
\cite{DFKVV99}   & Exact SVD & $k$                                 & $mn\min\{m,n\}$ & $2$    \\
\hline
\cite{FSS13}   & Exact SVD & $O(k/\varepsilon^2)$                                 & $mn\min\{m,n\}$ & $1+\varepsilon$    \\
\hline

Theorem~\ref{fastkmeans}        & RS & $O(k \log(k)/\varepsilon^2)$ & $mnk/\varepsilon $ & $3+\varepsilon$ \\
\hline
Theorem~\ref{thm:second_result}     & RP  & $O(k / \varepsilon^2)$                & $mn \lceil \varepsilon^{-2} k/ \log(n) \rceil$   & $2+\varepsilon$ \\
\hline
Theorem~\ref{thm:first_result}    & Approx. SVD  & $k$                & $mnk/\varepsilon$   & $2+\varepsilon$\\
\hline
\end{tabular}
\end{center}
\caption{\footnotesize
Provably accurate dimensionality reduction methods for $k$-means clustering. RP stands for Random Projections, and RS stands for Random Sampling. The third column corresponds to the number of selected/extracted features; the fourth column corresponds to the time complexity of each dimensionality reduction method; the fifth column corresponds to the approximation ratio of each approach. 
}
\label{tb:sum}
\end{table}

\subsection{Prior Work}\label{sec:prior}

Despite the significance of dimensionality reduction in the context of clustering, as well as the wealth of heuristic methods addressing it~\cite{GE03}, to the best of our knowledge there are no provably accurate feature selection methods for $k$-means clustering known. On the other hand, two provably accurate feature extraction methods are known in the literature that we describe next.

First, a result by~\cite{JL84} indicates that one can
construct $r = O( \log(m) / \varepsilon^{2})$ artificial features with Random Projections and,
with high probability, obtain a $(1+\varepsilon)$-approximate clustering. The algorithm implied by~\cite{JL84}, which is a random-projection-type algorithm,  is as follows: 
let $\matA \in \R^{m \times n}$ contain the points
$\mathcal{P} = \{ \p_1, \p_2, \dots, \p_m \} \subseteq \R^n$ as its rows; then, multiply $\matA$ from the right with
a random projection matrix $\matR \in \R^{n \times r}$ to construct $\matC = \matA \matR \in \R^{m \times r}$
containing the points $\mathcal{\hat{P}} = \{ \hat{\p}_1, \hat{\p}_2, \dots, \hat{\p}_m \} \subseteq \R^r$ as its rows (see Section~\ref{chap319} for a definition of a random projection matrix). The proof of this result is immediate mainly due to the Johnson-Lindenstrauss lemma~\cite{JL84}. \cite{JL84} proved
that all the pairwise Euclidean distances of the points of $\mathcal{P}$ are preserved within a multiplicative
factor $1 \pm \varepsilon$. So, any value of the $k$-means objective function, which depends only on pairwise
distances of the points  from the corresponding center point, is preserved within a factor $1 \pm \varepsilon$
in the reduced space. 
Second, \cite{DFKVV99} argues that one can construct $r = k$ artificial features using the SVD,
in $O( mn \min\{ m,n \} )$ time, to obtain a $2$-approximation on the clustering quality.
The algorithm of~\cite{DFKVV99} is as follows: given $\matA \in \R^{m \times n}$ containing the points of
$\mathcal{P}$ and $k$, construct $\matC = \matA \matV_{k} \in \R^{m \times k}$. Here,
$\matV_{k} \in \R^{n \times k}$ contains the top $k$ right singular vectors of $\matA$.
The proof of this result will be (briefly) discussed in Sections~\ref{sec:in} and~\ref{sec6}.

Finally, an extension of the latter SVD-type result (see Corollary 4.5 in~\cite{FSS13}) argues that $O(k/\varepsilon^2)$ dimensions (singular vectors) suffice for a relative-error approximation. 

\subsection{Summary of our Contributions}
We present the first provably accurate feature selection algorithm
for $k$-means (Algorithm~\ref{alg:featureSelect}). Namely, Theorem \ref{fastkmeans} presents
an $O( mnk\varepsilon^{-1} + k \log(k) \varepsilon^{-2} \log(k \log(k) \varepsilon^{-1}))$ time
randomized algorithm that, with constant probability, achieves a
$(3+\varepsilon)$-error with $r = O( k \log(k) /\varepsilon^{2} )$ features. Given $\matA$ and $k$,
the algorithm of this theorem first computes $\matZ \in \R^{n \times k}$, which approximates
$\matV_k \in \R^{n \times k}$ which contains the top $k$ right singular vectors of $\matA$
\footnote[2]{\cite{BMD09c} presented an unsupervised feature selection algorithm by working with the matrix
$\matV_k$; in this work, we show that the same approximation bound can be achieved by working with a matrix that
approximates $\matV_k$ in the sense of low rank matrix approximations (see Lemma~\ref{tropp2}).}.
Then, the selection of the features (columns of $\matA$)
is done with a standard randomized sampling approach with replacement with probabilities that are computed
from the matrix $\matZ$. The proof of Theorem \ref{fastkmeans} is a synthesis of ideas from~\cite{DFKVV99} and~\cite{RV07},
which study the paradigm of dimensionality reduction for $k$-means clustering and the paradigm of randomized sampling, respectively. 

Moreover, we describe a random-projection-type feature extraction algorithm:
Theorem~\ref{thm:second_result} presents an $O(m n \lceil \varepsilon^{-2} k / \log(n) \rceil )$ time
algorithm that, with constant probability,
achieves a $(2+\varepsilon)$-error with $r = O( k /\varepsilon^{2})$ artificial features.
We improve the folklore result of the first row in Table~\ref{tb:sum} by means of showing that a
smaller number of features are enough to obtain
an approximate clustering. The algorithm of Theorem \ref{thm:second_result} is the same
as with the one in the standard result for random projections that we outlined in the prior work section but uses only $r = O( k / \varepsilon^{2})$ dimensions for the random projection matrix. Our proof relies on ideas from~\cite{DFKVV99} and~\cite{Sar06}, which study the paradigm of dimension reduction for $k$-means clustering and the paradigm of speeding up linear algebra problems, such as the low-rank matrix approximation problem, via random projections, respectively. 

Finally, Theorem \ref{thm:first_result} describes a feature extraction algorithm
that employs approximate SVD decompositions and constructs $r = k$ artificial features
in $O(m n k /\varepsilon)$ time such that, with constant probability, the clustering
error is at most a $2+\varepsilon$ multiplicative factor from the optimal.
We improve the existing SVD dimensionality reduction method by showing that
fast approximate SVD gives features that can do almost as well
as the features from the exact SVD.
Our algorithm and proof are similar to those in~\cite{DFKVV99}, but
we show that one only needs to compute an approximate SVD of $\matA$.

We summarize previous results as well as our results in Table~\ref{tb:sum}.

\section{Linear Algebraic Formulation and our Approach}
\subsection{Linear Algebraic Formulation of $k$-means}\label{sec:problem}
From now on, we will switch to a 
linear algebraic formulation of the $k$-means clustering problem following the notation used in the introduction.
Define the data matrix
\math{\matA \in \R^{m \times n}} whose rows correspond to the data points,
$$\matA\transp=[\p_1,\ldots,\p_m] \in \R^{n \times m}.$$
We represent a $k$-clustering \math{\cl{S}} of $\matA$ by its \emph{cluster indicator matrix} $\matX \in \R^{m \times k}$.
Each column \math{j=1,\ldots,k} of \math{\matX} corresponds to a cluster.
Each row \math{i=1,\ldots,m}
indicates the
cluster membership of the point \math{\p_i \in \R^m}. So,
\math{\matX_{ij}=1/\sqrt{s_j}} if and
only if data point \math{\p_i} is in cluster \math{S_j}. Every row of $\matX$ has
exactly one non-zero element, corresponding to the cluster the data
point belongs to. There are \math{s_j} non-zero elements in column
\math{j} which indicates the data points belonging to cluster
\math{S_j}. By slightly abusing notation, we define 
$$\mathcal{F}(\matA,\matX) := \FNormS{\matA - \matX \matX\transp \matA}.$$ Hence, for any cluster indicator matrix $\matX$, the following identities hold
\[
\cl F(\matA,\matX)
= \sum_{i=1}^{m} \norm{ \p_i\transp - \matX_{i} \matX\transp\matA }_2^2
= \sum_{i=1}^m\norm{\p_i\transp-\bm\mu(\p_i)\transp}_2^2 
= \mathcal{F}(\mathcal{P}, \cl S),
\]
where we define \math{\matX_i} as the \math{i}th row of \math{\matX} and
we have used the identity $\matX_{i} \matX\transp\matA=\bm\mu(p_i)\transp,$ for $i=1,...,m$.
This identity is true because \math{\matX\transp\matA} is a matrix whose row \math{j}
is \math{\sqrt{s_j}\bm\mu_j}, proportional to the centroid of the
\math{j}th cluster; now, \math{\matX_{i}} picks the row corresponding to its
non-zero element, i.e., the cluster corresponding to point \math{i}, and scales
it by \math{1/\sqrt{s_j}}.
In the above, $\mu(\p_i) \in \R^m$ denotes the centroid of the cluster of which the point $\p_i$ belongs to. 
Using this formulation, the goal of $k$-means is to find
$\matX$ which minimizes \math{\FNormS{\matA - \matX \matX\transp \matA}}.

To evaluate the quality of different clusterings,
we will use the \math{k}-means objective function. Given some
clustering \math{\hat\matX}, we are interested
in the ratio $\cl{F}(\matA,\hat\matX)/\cl{F}(\matA,\matX_{\mathrm{opt}}),$
where \math{\matX_{\mathrm{opt}}} is an optimal clustering of $\matA$. The choice
of evaluating a clustering under this framework is not new. In fact,~\cite{OR99,KSS04,HM04,PK05,FS06,ORSS06}
provide results (other than dimensionality reduction methods) along the same lines.
Below, we give the definition of the $k$-means problem.
\begin{definition}
\textsc{[The $k$-means clustering problem]}
Given $\matA \in \mathbb{R}^{m \times n}$ (representing $m$ data
points -- rows -- described with respect to $n$ features --
columns) and a positive integer $k$ denoting the number of
clusters, find the indicator matrix $\matX_{\mathrm{opt}} \in \R^{m \times k}$ which satisfies,
$$
\matX_{\mathrm{opt}} = \argmin_{\matX \in \cal{X}} \FNormS{\matA - \matX \matX\transp \matA}.
$$
The optimal value of the $k$-means clustering objective is
$$
\cl{F}(\matA,\matX_{\mathrm{opt}}) = \min_{\matX \in \cal{X}} \FNormS{\matA - \matX \matX\transp \matA} = \FNormS{\matA - \matX_{\mathrm{opt}} \matX_{\mathrm{opt}}\transp\matA} = \mathrm{F}_{\mathrm{opt}}.$$
In the above, $\cal{X}$ denotes the set of all $m \times k$ indicator matrices $\matX$.
\end{definition}

Next, we formalize the notation of a ``$k$-means approximation algorithm''.
\begin{definition} \label{def:approx}
\textsc{[$k$-means approximation algorithm]}
An algorithm is called a ``$\gamma$-approximation'' for the $k$-means clustering problem ($\gamma \geq 1$) if
it takes inputs the dataset $\matA \in \R^{m \times n}$ and the number of clusters $k$, and returns an indicator matrix
$\matX_{\gamma} \in \R^{m \times k}$ such that
w.p. at least $1 - \delta_{\gamma}$,
$$
\FNormS{\matA - \matX_{\gamma} \matX_{\gamma}\transp \matA} \leq \gamma \min_{\matX \in
\cal{X}} \FNormS{\matA - \matX \matX\transp \matA} = \gamma \cdot \cl{F}(\matA,\matX_{\mathrm{opt}}) = \gamma \cdot \mathrm{F}_{\mathrm{opt}}.
$$
\end{definition}
\vspace{-.1in}
\noindent
An example of such an approximation algorithm for $k$-means is in~\cite{KSS04}
with $\gamma = 1+\varepsilon$ ($0 < \varepsilon < 1$),
and $\delta_{\gamma}$ a constant in $(0,1)$.
The corresponding running time is $O(m n \cdot 2^{(k/\varepsilon)^{O(1)}})$.

Combining this algorithm (with $\gamma = 1+\varepsilon$) with, for example, our dimensionality reduction method in Section~\ref{sec5}, would result in an algorithm that preserves the clustering within a factor of 
$2 + \varepsilon$, for any $\varepsilon \in (0,1/3)$, and runs in total time 
$O( m n \lceil \varepsilon^{-2} k/ \log(n) \rceil + k n 2^{(k/\varepsilon)^{O(1)}} / \varepsilon^2 )$. Compare this with the complexity of running this algorithm on the high dimensional data and notice that reducing the dimension from $n$ to $O(k/\varepsilon^2)$ leads to a considerably faster algorithm.
In practice though, the Lloyd
algorithm~\cite{Llo82,ORSS06} is very popular and although it does not admit a 
worst case theoretical analysis, it empirically does well. 
We thus employ the Lloyd algorithm for our experimental evaluation of our algorithms in Section~\ref{sec:experiments}. Note that, after using, for example, the dimensionality reduction method in 
Section~\ref{sec5}, the cost of the Lloyd heuristic is only $O( m k^2 /\varepsilon^2 )$ per iteration. 
This should be compared to the cost of  $O( k m n )$ per iteration if applied on the original high dimensional data. Similar run time improvements arise if one uses the other dimension reduction algorithms proposed in this work. 

\subsection{Our Approach}\label{sec:in}
The key insight of our work is to view the $k$-means problem from the above linear algebraic perspective. In this setting, the data points are rows in a matrix $\matA$ and
feature selection corresponds to selection of a subset of columns from $\matA$. Also, feature extraction can be viewed as the construction of a matrix $\matC$ which contains the constructed features. Our feature extraction algorithms are linear, i.e., the matrix $\matC$ is of the form $\matC = \matA \matD$, for some matrix
$\matD$; so, the columns in $\matC$ are linear combinations of the columns of $\matA$,
i.e., the new features are linear combinations of the original features.

Our work is inspired by the SVD feature extraction algorithm of~\cite{DFKVV99},
which also viewed the $k$-means problem from a linear algebraic perspective. The main
message of the result of~\cite{DFKVV99} (see the algorithm and the analysis in Section 2 in~\cite{DFKVV99}) 
is that any matrix $\matC$ which can be used
to approximate the matrix $\matA$ in some low-rank matrix approximation sense
can also be used for dimensionality reduction in $k$-means clustering. 
We will now present a short proof of this result to better understand its implications in our
dimensionality reduction algorithms. 

Given $\matA$ and $k$, the main algorithm of~\cite{DFKVV99} constructs
$\matC = \matA \matV_k$, where $\matV_k$ contains the top $k$ right singular vectors of $\matA$. Let $\matX_{opt}$ and $\hat{\matX}_{opt}$ be the cluster indicator matrices that corresponds to the optimum partition corresponding to the rows of $\matA$
and  the rows of $\matC$, respectively. In our setting for dimensionality reduction,
we compare $\cl{F}(\matA,\hat{\matX}_{opt})$ to $\cl{F}(\matA,\matX_{opt})$. From the SVD of $\matA$, consider
$$
\matA =  \underbrace{\matA\matV_k\matV_k\transp}_{\matA_k} + \underbrace{\matA - \matA\matV_k\matV_k\transp}_{\matA_{\rho-k}}.
$$
Also, notice that for any cluster indicator matrix $\hat{\matX}_{opt}$
\[
\left( \left( \matI_{m } - \hat{\matX}_{opt}\hat{\matX}_{opt}\transp\right) \matA_k \right)
\left( \left( \matI_{m } - \hat{\matX}_{opt}\hat{\matX}_{opt}\transp\right) \matA_{\rho-k} \right)\transp
= {\bf 0}_{m \times m},
\]
because $\matA_k \matA_{\rho-k}\transp = {\bf 0}_{m \times m}$. Combining these two steps and by orthogonality, it follows that
\[
\FNorm{\matA - \hat{\matX}_{opt} \hat{\matX}_{opt}\transp \matA}^2\ =\ \underbrace{\FNorm{(\matI_m -
\hat{\matX}_{opt} \hat{\matX}_{opt}\transp ) \matA_k}^2}_{\theta_{\alpha}^2}\ +\ \underbrace{\FNorm{(\matI_m - \hat{\matX}_{opt} \hat{\matX}_{opt}\transp )\matA_{\rho-k}}^2}_{\theta_{\beta}^2}.
\]
We now bound the second term of the later equation. $\matI_m-\hat{\matX}_{opt} \hat{\matX}_{opt}\transp $ is a projection
matrix, so it can be dropped without increasing the Frobenius norm. Hence, by using this and the fact that $\matX_{\mathrm{opt}}\matX_{\mathrm{opt}}\transp \matA$ has rank at most $k$:
$$
\theta_{\beta}^2\ \leq\ \FNorm{\matA_{\rho-k}}^2\  = \FNorm{\matA - \matA_k}^2\ \leq\ \FNormS{ \matA -
\matX_{\mathrm{opt}}\matX_{\mathrm{opt}}\transp \matA }.
$$
From similar manipulations combined with the optimality of $\hat{\matX}_{opt}$, it follows that 
$$ \theta_{\alpha}^2\ \leq\ \FNormS{ \matA -
\matX_{\mathrm{opt}}\matX_{\mathrm{opt}}\transp \matA }.
$$ 
Therefore, we conclude that $\cl{F}(\matA,\hat{\matX}_{opt}) \leq 2 \cl{F}(\matA,\matX_{opt})$.
The key insight in this approach is that $\matA_k = \matA \matV_k \matV_k\transp = \matC \cdot \matH$ (with $\matH = \matV_k\transp$)
and $\matA - \matC\matH  = \matA_{\rho-k}$, which is the best rank $k$ approximation of $\matA$ in the Frobenius norm
(see Section~\ref{sec2} for useful notation).

In all three methods of our work, we will construct matrices $\matC = \matA \matD$, for three different matrices
$\matD$, such that $\matC \cdot \matH$, for an appropriate $\matH$,
is a good approximation to $\matA$ with respect to the Frobenius norm,
i.e., $\FNormS{\matA - \matC \cdot \matH}$ is roughly equal to 
$\FNormS{\matA - \matA_k},$ where $\matA_k$ is the best rank $k$ matrix from the SVD of $\matA$.  Then, replicating the above proof gives our main theorems. Notice that the above
approach is a $2$-approximation because $\matA_k = \matC \cdot \matH$ is the best rank $k$ approximation to $\matA$; our algorithms
will give a slightly worse error because our matrix $\matC \cdot \matH$ give an approximation which is slightly worse than
the best rank $k$ approximation. 

%
%
\section{Preliminaries}\label{sec2}
\paragraph{Basic Notation.}  \label{chap21}
We use \math{\matA,\matB,\dots} to denote matrices;
\math{\a, \p,  \dots} to denote column vectors.
$\matI_{n}$ is the $n \times n$
identity matrix;  $\bm{0}_{m \times n}$ is the $m \times n$ matrix of zeros;
$\matA_{(i)}$ is the $i$-th row of $\matA$; $\matA^{(j)}$ is the $j$-th column of $\matA$; and,
$\matA_{ij}$ denotes the $(i,j)$-th element of $\matA$.
We use $\Expect{Y}$ to take the expectation of a random variable $Y$
and $\Prob({\cal E})$ to take the probability of a probabilistic event
${\cal E}$. We abbreviate ``independent identically
distributed'' to ``i.i.d'' and ``with probability'' to ``w.p''.
\paragraph{Matrix norms.}
 \label{chap23}
We use the Frobenius and the spectral matrix norms:
$ \FNorm{\matA} = \sqrt{\sum_{i,j} \matA_{ij}^2}$ and
$\TNorm{\matA} = \max_{\x:\TNorm{\x}=1}\TNorm{\matA \x}$, respectively (for a matrix $\matA$).
For any $\matA$,$\matB$: $\TNorm{\matA}\le\FNorm{\matA}$, $\FNorm{\matA\matB} \leq \FNorm{\matA}\TNorm{\matB}$, and $ \FNorm{\matA\matB} \leq \TNorm{\matA} \FNorm{\matB}$. The latter two properties are stronger versions of the standard submultiplicativity property:
$\FNorm{\matA\matB} \leq \FNorm{\matA}\FNorm{\matB}$. We will refer to these versions
as spectral submultiplicativity. Finally, the triangle inequality of matrix norms indicates that
$\FNorm{\matA + \matB} \le \FNorm{\matA}+\FNorm{\matB}$.
\begin{lemma}[Matrix Pythagorean Theorem]\label{lem:pythagoras}
Let \math{\matX,\matY\in\R^{m\times n}} satisfy \math{\matX\matY\transp=\bm{0}_{m \times m}}. Then,
\[
\FNorm{\matX+\matY}^2 = \FNorm{\matX}^2+\FNorm{\matY}^2.
\]
\end{lemma}
\begin{proof} \begin{eqnarray*}
\FNorm{\matX+\matY}^2 &=& \Trace{ \left(\matX+\matY\right)\left(\matX+\matY\right)\transp } =
\Trace{ \matX\matX\transp + \matX\matY\transp + \matY\matX\transp + \matY\matY\transp }  \\
&=& \Trace{ \matX\matX\transp + \bm{0}_{m \times m} + \bm{0}_{m \times m} + \matY\matY\transp } =
\Trace{ \matX\transp\matX} +  \Trace{\matY\matY\transp } \\
&=& \FNorm{\matX}^2+\FNorm{\matY}^2. \end{eqnarray*}
\end{proof}
\noindent This matrix form of the Pythagorean theorem is the starting point for the
proofs of the three main theorems presented in this work. The idea to use the Matrix Pythagorean theorem to analyze
a dimensionality reduction method for $k$-means
was initially introduced in~\cite{DFKVV99} and it turns to be very useful to prove our results as well.
\paragraph{Singular Value Decomposition.}
The SVD of $\matA \in \R^{m \times n}$ of rank $\rho \le \min\{m,n\}$ is $\matA = \matU_{\matA}\matSig_{\matA}\matV_\matA\transp$,
with $\matU_{\matA} \in \R^{m \times \rho},$
$\matSig_{\matA} \in \R^{\rho \times \rho},$ and
$\matV_\matA \in \R^{n \times \rho}$. 
In some more details, the SVD of $\matA$ is: 
\begin{eqnarray*}
\label{svdA} \matA
         = \underbrace{\left(\begin{array}{cc}
             \matU_{k} & \matU_{\rho-k}
          \end{array}
    \right)}_{\matU_{\matA} \in \R^{m \times \rho}}
    \underbrace{\left(\begin{array}{cc}
             \matSig_{k} & \bf{0}\\
             \bf{0} & \matSig_{\rho - k}
          \end{array}
    \right)}_{\matSig_\matA \in \R^{\rho \times \rho}}
    \underbrace{\left(\begin{array}{c}
             \matV_{k}\transp\\
             \matV_{\rho-k}\transp
          \end{array}
    \right)}_{\matV_\matA\transp \in \R^{\rho \times n}},
\end{eqnarray*}
with singular values \math{\sigma_1\ge\ldots\geq \sigma_k\geq\sigma_{k+1}\ge\ldots\ge\sigma_\rho > 0}. We
will use $\sigma_i\left(\matA\right)$ to denote the $i$-th singular value of $\matA$ when the matrix is not clear from the context.
The matrices
$\matU_k \in \R^{m \times k}$ and $\matU_{\rho-k} \in \R^{m \times (\rho-k)}$ contain the left singular vectors of~$\matA$; and, similarly, the matrices $\matV_k \in \R^{n \times k}$ and $\matV_{\rho-k} \in \R^{n \times (\rho-k)}$ contain the right singular vectors.
$\matSig_k \in \R^{k \times k}$ and $\matSig_{\rho-k} \in \R^{(\rho-k) \times (\rho-k)}$ contain the singular values of~$\matA$.
It is well-known that $\matA_k=\matU_k \matSig_k \matV_k\transp = \matA \matV_k\matV_k\transp = \matU_k\matU_k\transp\matA$ minimizes \math{\FNorm{\matA - \matX}} over all
matrices \math{\matX \in \R^{m \times n}} of rank at most $k \le \rho$. We use $\matA_{\rho-k} = \matA - \matA_k = \matU_{\rho-k}\matSig_{\rho-k}\matV_{\rho-k}\transp$. Also, $ \FNorm{\matA} =
\sqrt{ \sum_{i=1}^\rho\sigma_i^2(\matA) }$ and $\TNorm{\matA} = \sigma_1(\matA)$. The best rank $k$ approximation to $\matA$ also satisfies:
$\FNorm{\matA-\matA_k} = \sqrt{\sum_{i=k+1}^{\rho}\sigma_{i}^2(\matA)}$.
\paragraph{Approximate Singular Value Decomposition.}
The exact SVD of $\matA$ takes cubic time. In this work, to speed up certain algorithms,
we will use fast approximate SVD. We quote a recent result from~\cite{BDM11a}, but similar
relative-error Frobenius norm SVD approximations can be found elsewhere; see, for example,~\cite{Sar06}.
%
\begin{lemma}
\label{tropp2}
Given \math{\matA\in\R^{m\times n}} of rank $\rho$, a target rank $2\leq k < \rho$, and $0 < \varepsilon < 1$, there exists a randomized algorithm that  computes a matrix $\matZ \in \R^{n \times k}$ such that
$\matZ\transp\matZ = \matI_k$, \math{\matE\matZ=\bm{0}_{m \times k}} (for $\matE = \matA - \matA \matZ \matZ\transp \in \R^{m \times n}$), and
$$\Expect{\FNormS{\matE}} \leq \left(1+{\varepsilon}\right)\FNormS{\matA - \matA_k}.$$
The proposed algorithm runs in $O\left(mnk/\varepsilon\right)$ time.
We use $\matZ = \text{FastFrobeniusSVD}(\matA, k, \varepsilon)$ to denote this
algorithm.
\end{lemma}
Notice that this lemma computes a rank-$k$ matrix $\matA\matZ\matZ\transp$ which, when is used to approximate $\matA$, 
is almost as good - in expectation -  as the rank-$k$ matrix $\matA_k$ from the SVD of $\matA$.
Since, $\matA_k = \matA \matV_k \matV_k\transp$, the matrix $\matZ$ is essentially an approximation of the matrix $\matV_k$ from the SVD of $\matA$. 

We now give the details of the algorithm. The algorithm takes as inputs a matrix $\matA \in \R^{m \times n}$ of rank $\rho$ and 
an integer $2 \leq k < \rho$. Set $r = k +  \ceil{ \frac{k}{\varepsilon } + 1}$ and construct $\matZ$ with the following algorithm.
\begin{center}
\begin{algorithmic}[1]
\STATE  Generate an $n \times r$ standard Gaussian matrix $\matR$ whose entries are i.i.d. $\mathcal{N}(0,1)$ variables.
\STATE $\matY = \matA \matR \in \R^{m \times r}$.
\STATE Orthonormalize the columns of $\matY$ to construct the matrix $\matQ \in \R^{m \times r}$.
\STATE Let $\matZ \in \R^{n \times k}$ be the top $k$ right singular vectors of $\matQ\transp \matA \in \R^{r \times n}$.
\end{algorithmic}
\end{center}
\paragraph{Pseudo-inverse.} $\matA^\dagger = \matV_\matA \matSig_\matA^{-1} \matU_\matA\transp \in \R^{n \times m}$
denotes the so-called Moore-Penrose pseudo-inverse of $\matA$ (here $\matSig_\matA^{-1}$ is the inverse of $\matSig_\matA$),
i.e., the unique $n \times m$ matrix satisfying all four properties:
$\matA = \matA \matA^\dagger \matA$,
$\matA^\dagger \matA \matA^\dagger = \matA^\dagger$,
$(\matA \matA^\dagger )\transp = \matA \matA^\dagger$, and
$(\matA^\dagger \matA )\transp = \matA^\dagger\matA$.
By the SVD of $\matA$ and $\matA^\dagger$, it is easy to verify
that, for all $i=1,\dots,\rho = \rank(\matA) = \rank(\matA^\dagger)$:
$\sigma_i(\matA^\dagger) = 1 / \sigma_{\rho - i + 1}(\matA)$.
Finally, for any $\matA \in \R^{m \times n}, \matB \in \R^{n \times \ell}$:
\math{(\matA\matB)^\dagger=\matB^\dagger\matA^\dagger} if any one
of the following three properties hold:
(\rn{1}) \math{\matA\transp\matA=\matI_n};
(\rn{2}) \math{\matB\transp\matB=\matI_{\ell}};
or, (\rn{3}) $\rank(\matA) = \rank(\matB) = n$.
\paragraph{Projection Matrices.}  \label{chap28} We call
$\matP \in \R^{n \times n}$ a projection matrix if
$\matP^2=\matP$. For such a projection matrix and any $\matA$: $\FNorm{\matP \matA} \leq \FNorm{\matA}.$
Also, if $\matP$ is a projection matrix, then, $\matI_{n} - \matP$ is a projection matrix. So,
for any matrix $\matA$, both $\matA \matA^\dagger$ and $\matI_{n} - \matA \matA^\dagger$ are projection matrices.

\paragraph{Markov's Inequality and the Union Bound.} \label{chap29}
Markov's inequality can be stated as follows:
Let $Y$ be a random variable taking non-negative values with
expectation $\Expect{Y}$. Then, for all $t > 0$, and with probability at least $1-t^{-1}$,
$Y \leq t \cdot \Expect{Y}.$
We will also use the so-called union bound. Given a set of probabilistic events ${\cal E}_1,{\cal E}_2,\ldots,{\cal E}_{n}$ holding with respective probabilities $p_1,p_2,\ldots,p_n$, the probability that all events hold simultaneously
(a.k.a., the probability of the union of those events) is upper bounded as:
$\Prob ({\cal E}_1 \cup {\cal E}_2 \ldots \cup {\cal E}_n)  \le \sum_{i=1}^n p_i. $
\subsection{Randomized Sampling }\label{chap314}
%

\paragraph{Sampling and Rescaling Matrices.}  \label{chap22} Let \math{\matA=[\matA^{(1)},\ldots,\matA^{(n)}] \in \R^{m \times n}} and let
$\matC=[\matA^{(i_1)},\ldots,\matA^{(i_r)}] \in \R^{m \times r}$ consist of \math{r<n}
columns of~\math{\matA}. Note that we can write \math{\matC=\matA\matOmega}, where the \emph{sampling matrix} is
\math{\matOmega=[\e_{i_1},\ldots,\e_{i_r}] \in \R^{n \times r}} (here \math{\e_i} are the standard basis vectors in \math{\R^n}).
If $\matS \in \R^{r \times r}$ is a diagonal \emph{rescaling matrix} then $\matA \matOmega \matS$ contains $r$ rescaled columns of $\matA$.

The following definition describes a simple randomized sampling procedure with replacement, which will be critical in our feature selection algorithm.

\begin{definition}[Random Sampling with Replacement] \label{def:sampling}
Let $\matX \in \R^{n \times k}$ with $n > k$ and let $\matX_{(i)}$
denote the $i$-th row of $\matX$ as a row vector.
For all $i=1,\dots ,n,$ define the following set of sampling probabilities: $$p_i 
= \frac{\TNormS{\matX_{(i)}}}{\FNormS{\matX}},$$
and note that $ \sum_{i=1}^{n} p_i = 1$. Let $r$ be a positive integer and construct the sampling matrix $\matOmega \in \R^{n \times r}$ and the rescaling matrix $\matS \in \R^{r \times r}$ as follows: initially, $\matOmega = \bm{0}_{n \times r}$ and $\matS=\bm{0}_{r \times r}$; for $t=1,\dots ,r$ pick an integer $i_t$ from the set $\{1,2,\dots ,n\}$ where the probability of picking $i$ is equal to $p_i$; set $\matOmega_{i_tt} = 1$ and $\matS_{tt} = 1/\sqrt{rp_{i_t}}$. We denote this
randomized sampling technique with replacement by
\[[\matOmega, \matS] = \text{RandomizedSampling}(\matX, r).\]
Given $\matX$ and $r$, it takes $O(nk)$ time to compute the probabilities and another  
$O(n + r)$ time to implement the sampling procedure via the technique in~\cite{Vos91}. In total, this method
requires $O(nk)$ time.  
\end{definition}

The next three lemmas present the effect of the above sampling procedure on certain spectral properties, e.g. singular values, of orthogonal matrices. 
The first two lemmas are known; short proofs are included for the sake of completeness. The third lemma follows easily from the first two results (a proof of the lemma is given for completeness as well). We remind the reader that 
$\sigma_i^2(\matX)$ denotes the $i$th singular value squared of the matrix $\matX$. 

Lemma~\ref{lem:random} argues that sampling and rescaling a sufficiently large number of rows from an orthonormal matrix with the randomized procedure of Definition~\ref{def:sampling} results in a matrix with singular values close to the singular values of the original orthonormal matrix. 
\begin{lemma}
\label{lem:random}
Let $\matV \in \R^{n \times k}$ with $n > k$ and $\matV\transp \matV = \matI_{k}$.
Let $0 < \delta < 1$, $ 4 k \ln( 2 k / \delta)  < r \leq n$, and
$[\matOmega, \matS] = \text{RandomizedSampling}(\matV,  r)$.
Then, for all $i=1,\dots,k$, w.p. at least $1 - \delta$,
$$1 -  \sqrt{\frac{4 k \ln(2 k / \delta )}{ r }}  \leq  \sigma_i^2(\matV\transp \matOmega \matS)
\leq 1+ \sqrt{\frac{4  k \ln(2k/\delta)}{ r }}.
$$
\end{lemma}
\begin{proof} This result was originally proven in \cite{RV07}. We will leverage a more recent proof of this result that appeared in~\cite{Mag10} and improves the original constants.
More specifically, in Theorem 2 of~\cite{Mag10}, set $\matS = \matI$, $\beta=1$,  and replace $\varepsilon$ as a function of $r$, $\beta$, and $d$ to conclude the proof.
\end{proof}
Lemma~\ref{lem:fnorm} argues that sampling and rescaling columns from any matrix with the randomized procedure of 
Definition~\ref{def:sampling} results in a matrix with Frobenius norm squared close to the Frobenius norm squared of the original matrix. Intuitively, the subsampling of the columns does not affect much the Frobenius norm of the matrix. 

\begin{lemma} \label{lem:fnorm}
For any $r \ge 1$, $\matX \in \R^{n \times k}$, and $\matY \in \R^{m \times n}$, let
$[\matOmega, \matS] = \text{RandomizedSampling}(\matX,  r)$. Let $\delta$ be a parameter with $0 < \delta < 1$. Then, w.p. at least $1-\delta$,
$$ \FNormS{ \matY \matOmega \matS } \leq \frac{1}{\delta} \FNormS{ \matY }. $$
\end{lemma}
\begin{proof}
\noindent Define the random variable
$Y = \FNormS{ \matY \matOmega \matS } \ge 0$. Assume that the following equation is true:
$ \Expect{ \FNormS{ \matY \matOmega \matS } } = \FNormS{ \matY }.$
Applying Markov's inequality 	with failure probability $\delta$ to this equation gives the bound in the lemma.
All that remains to prove now is the above assumption.
Let $\matB = \matY \matOmega \matS \in \R^{m \times r},$ and for $t=1, \dots ,r,$ let $\matB^{(t)}$
denotes the $t$-th column of $\matB = \matY \matOmega \matS$. We manipulate
the term $\Expect{ \FNormS{ \matY \matOmega \matS } }$ as follows,
\eqan{
\Expect{ \FNormS{ \matY \matOmega \matS } } \buildrel{(a)}\over{=}  \Expect{ \sum_{t=1}^{r} \TNormS{\matB^{(t)}} } \buildrel{(b)}\over{=}
\sum_{t=1}^{r} \Expect{ \TNormS{\matB^{(t)}} } \buildrel{(c)}\over{=}
\sum_{t=1}^{r} \sum_{j=1}^n p_j \frac{ \TNormS{\matY^{(j)}} }{r p_j} \buildrel{(d)}\over{=}
\frac{1}{r} \sum_{t=1}^{r} \FNormS{\matY} = \FNormS{\matY}
}
\math{(a)} follows by the definition of the Frobenius norm of $\matB$.
\math{(b)} follows by the linearity of expectation.
\math{(c)} follows by our construction of $\matOmega, \matS$.
\math{(d)} follows by the definition of the Frobenius norm of $\matY$. It is worth noting
that the above manipulations hold for any set of probabilities since they cancel out in Equation \math{(d)}.
\end{proof}
Notice that $\matX$ does not appear in the bound; it is only used as an input to the \text{RandomizedSampling}. This means
that for \emph{any} set of probabilities, a sampling and rescaling matrix constructed in the way it is described in Definition~\ref{def:sampling}
satisfies the bound in the lemma. 

The next lemma shows the effect of sub-sampling in a low-rank approximation of the form $\matA \approx \matA \matZ \matZ\transp,$ where $\matZ$ is a tall-and-skinny orthonormal matrix. The sub-sampling here is done on the columns of $\matA$ and the corresponding rows of $\matZ$. 

\begin{lemma}\label{lem:rsall}
Fix $\matA\in\R^{m\times n}$, $k \ge 1$, $0 < \varepsilon < 1/3$,  $0 < \delta < 1$, and $r = 4 k \ln(2k/\delta)/\varepsilon^2$. Compute the $n\times k$ matrix $\matZ$ of Lemma~\ref{tropp2} such that $\matA = \matA\matZ\matZ\transp + \matE$ and run $[\matOmega, \matS] = \text{RandomizedSampling}(\matZ, r)$.  Then, w.p. at least $1-3\delta$, there exists $\widetilde{\matE} \in \R^{m \times n}$ such that $$ \matA \matZ \matZ\transp = \matA \matOmega \matS (\matZ\transp \matOmega \matS)^\dagger \matZ\transp + \widetilde{\matE},$$
and $\FNorm{ \widetilde{\matE} } \le \frac{ 1.6 \varepsilon}{\sqrt{\delta}}  \FNorm{\matE}.$
\end{lemma}
\begin{proof} See Appendix.
\end{proof}
In words, given $\matA$ and the rank parameter $k$, it is possible to construct two low rank matrices, $\matA \matZ \matZ\transp$ and $\matA \matOmega \matS (\matZ\transp \matOmega \matS)^\dagger \matZ\transp$ that are ``close'' to each other. Another way to view this result is that given the low-rank factorization $\matA \matZ \matZ\transp$ one can ``compress'' $\matA$ and $\matZ$ by means of the sampling and rescaling matrices $\matOmega$ and $\matS.$ The error from such a compression will be bounded by $\tilde{\matE}$.

This result is useful in proving Theorem~\ref{fastkmeans} because at some point of the proof (see Eqn.~(\ref{t0cc})) we need to switch from a rank $r$ matrix ($\matA \matOmega \matS (\matZ\transp \matOmega \matS)^\dagger \matZ\transp$) to a rank $k$ matrix ($\matA \matZ \matZ\transp$) and at the same time keep the bounds in the resulting inequality almost unchanged (they would change by the norm of the matrix $\widetilde{\matE}$). 

\subsection{Random Projections}\label{chap319}
A classic result of~\cite{JL84} states that,
for any $ 0 < \varepsilon < 1$, any set of $m$ points in $n$ dimensions
(rows in $\matA \in \R^{m \times n}$)
can be linearly projected into 
$$r_{\varepsilon}=O\left(\log (m) /\varepsilon^2\right)$$ dimensions
while preserving all the pairwise Euclidean distances of the points within a multiplicative factor of $(1\pm\varepsilon)$.
More precisely,~\cite{JL84} showed the existence of a (random orthonormal) matrix
$\matR \in \R^{n \times r_{\varepsilon}}$ such that,
for all $i,j=1,\dots ,m$, and with high probability (over the randomness of the matrix $\matR$),
\[
(1- \varepsilon) \norm{ \matA_{(i)} - \matA_{(j)} }_2
\leq  \norm{ \left(\matA_{(i)} - \matA_{(j)}\right)\matR }_2 \leq
(1 + \varepsilon) \norm{ \matA_{(i)} - \matA_{(j)}}_2.
\]
Subsequent research simplified the proof of~\cite{JL84} by showing that such a linear transformation
can be generated using a random Gaussian matrix, i.e., a matrix $\matR \in \R^{n \times r_{\varepsilon}}$
whose entries are i.i.d. Gaussian random variables
with zero mean and variance $1/r$~\cite{IM98}. Recently,~\cite{AC06}
presented the so-called Fast Johnson-Lindenstrauss Transform which describes
an  $\matR \in \R^{n \times r_{\varepsilon}}$
such that the product $\matA \matR$ can be computed fast. In this paper, we will use
a construction by \cite{Ach03}, who proved that a rescaled random sign matrix,
i.e., a matrix  $\matR \in \R^{n \times r_{\varepsilon}}$ whose entries
have values $\{\pm1/\sqrt{r}\}$ uniformly at random,
satisfies the above equation. As we will see in detail in Section~\ref{sec5},
a recent result of~\cite{LZ09} indicates that, if $\matR$ is constructed as in~\cite{Ach03},
the product $\matA \matR$ can be computed fast as well.
We utilize such a random projection embedding in Section \ref{sec5}.
Here, we summarize some properties of such matrices that might be of independent interest.
We have deferred the proofs of the following lemmata to the Appendix.

The first lemma argues that the Frobenius norm squared of any matrix $\matY$ and the Frobenius norm squared of $\matY \matR,$ where $\matR$ is a scaled signed matrix, are ``comparable''. Lemma~\ref{lem:RP:FNorm} is the analog of Lemma~\ref{lem:fnorm}.
\begin{lemma}\label{lem:RP:FNorm}
Fix any $m\times n$ matrix $\matY$, fix $k>1$ and $\varepsilon >0 $. Let $\matR \in \R^{n \times r}$ be a rescaled random sign matrix constructed as described above with $r = c_0 k /\varepsilon^{2}$, where $c_0\geq 100$. Then,
\begin{equation*}
	\Prob \left( \FNorm{\matY\matR}^2 \geq (1+\varepsilon) \FNorm{\matY}^2 \right) \leq 0.01.
\end{equation*}
\end{lemma}
The next lemma argues about the effect of scaled random signed matrices to the singular values of orthonormal matrices. 
\begin{lemma}\label{lem:rpall}
Let $\matA \in \R^{m \times n}$ with rank $\rho$ ($k < \rho$), $\matA_k = \matU_k \matSig_k \matV_k\transp$, and $0 < \varepsilon < 1/3$.
Let $\matR \in \R^{n \times r}$ be a (rescaled) random sign matrix constructed as we described above
with $r = c_0 k /\varepsilon^{2}$, where $c_0 \ge 3330$. The following hold (simultaneously) w.p. at least $0.97$:
\begin{enumerate}
\item For all $i=1,\dots ,k$:  
$$ 1 - \varepsilon \le \sigma_i^2(\matV_k\transp \matR) \le 1 +  \varepsilon.$$
\item There exists an $m\times n$ matrix $\widetilde{\matE}$ such that 
$$\matA_k  = \matA \matR (\matV_k\transp  \matR)^\dagger\matV_k\transp  + \widetilde{\matE},$$ and $$\FNorm{\widetilde{\matE}} \leq 3 \varepsilon \FNorm{\matA-\matA_k}.$$
\end{enumerate}
\end{lemma}
\noindent
 The first statement of
Lemma~\ref{lem:rpall} is the analog of Lemma~\ref{lem:random} while the second statement of
Lemma~\ref{lem:rpall} is the analog of Lemma~\ref{lem:rsall}. The results here replace the sampling
and rescaling matrices $\matOmega, \matS$ from Random Sampling (Definition~\ref{def:sampling}) with the Random Projection matrix $\matR$.
It is worth noting that almost the same results can be achieved with $r = O(k/\varepsilon^2)$ random dimensions,
while the corresponding lemmata for Random Sampling require at least $r = O(k \log k /\varepsilon^2)$ actual dimensions.

The second bound in the lemma is useful in proving Theorem~\ref{thm:second_result}. Specifically, in Eqn.~(\ref{thm:second_result}) we need to replace
the rank $k$ matrix $\matA_k$ with another matrix of rank $k$ which is as close to $\matA_k$ as possible. The second bound above provides precisely such a matrix
$\matA \matR (\matV_k\transp  \matR)^\dagger\matV_k\transp$ with corresponding error $\widetilde{\matE}$. 


\begin{algorithm}[t]
\begin{framed}
\textbf{Input:} Dataset $\matA\in\R^{m\times n}$, number of clusters $k$, and \math{0 < \varepsilon < 1/3}. \\
\noindent \textbf{Output:} \math{\matC \in \R^{m \times r}} with $r = O(k \log(k) / \varepsilon^2)$ rescaled features.
\begin{algorithmic}[1]
\STATE Let $ \matZ = \text{FastFrobeniusSVD}(\matA, k, \varepsilon)$; $\matZ \in \R^{n \times k}$ (via Lemma~\ref{tropp2}).
\STATE Let $r = c_1 \cdot 4 k \ln(200k)/\varepsilon^2$ ($c_1$ is a sufficiently large constant - see proof).
\STATE Let $[\matOmega, \matS] = \text{RandomizedSampling}(\matZ, r)$; $\matOmega \in \R^{n \times r}, \matS \in \R^{r \times r}$(via Lemma~\ref{lem:random}).
\STATE Return $\matC = \matA \matOmega \matS \in \R^{m \times r}$ with $r$ rescaled columns from $\matA$.
\end{algorithmic}
\caption{Randomized Feature Selection for $k$-means Clustering.}
\label{alg:featureSelect}
\end{framed}
\end{algorithm}

\section{Feature Selection with Randomized Sampling}\label{sec4}
Given $\matA, k$, and $0 < \varepsilon < 1/3$, Algorithm~\ref{alg:featureSelect}
is our main algorithm for feature selection in $k$-means clustering.
In a nutshell, construct the matrix $\matZ$ with the (approximate) top-$k$ right
singular vectors of $\matA$ and select $$r = O(k \log(k)/\varepsilon^2)$$ columns
from $\matZ\transp$ with the randomized technique of Section~\ref{chap314}.
One can replace the first step in Algorithm~\ref{alg:featureSelect} with the exact SVD of
$\matA$~\cite{BMD09c}. The result that is obtained from this approach is asymptotically the
same as the one we will present in Theorem~\ref{fastkmeans}
\footnote[3]{The main theorem of \cite{BMD09c} states a
$\left(1+(1+\varepsilon)\gamma\right)$-approximation bound but the corresponding proof has a bug,
which is fixable and leads to a $\left(1+(2+\varepsilon)\gamma\right)$-approximation bound.
One can replicate the corresponding (fixable) proof in \cite{BMD09c} by replacing $\matZ = \matV_k$ in the proof of
Theorem~\ref{fastkmeans} of our work.}.
Working with $\matZ$ though gives a considerably faster algorithm.
\begin{theorem}\label{fastkmeans}
Let $\matA \in \R^{m \times n}$ and $k$ be inputs of the $k$-means clustering problem.
Let $\varepsilon \in (0,1/3)$ and, by using Algorithm~\ref{alg:featureSelect}
in $O( m n k/\varepsilon + k \ln(k)/\varepsilon^2\log(k \ln(k)/\varepsilon) )$ time
construct features $\matC \in \R^{m \times r}$ with
$r = O(k \log(k) / \varepsilon^2)$. Run any
$\gamma$-approximation $k$-means algorithm with failure probability $\delta_{\gamma}$
on $\matC, k$ and construct $\matX_{\tilde{\gamma}}$.
Then, w.p. at least $0.2 - \delta_{\gamma}$,
$$
\FNorm{\matA - \matX_{\tilde{\gamma}} \matX_{\tilde{\gamma}}\transp  \matA}^2
\leq \left(1+(2+\varepsilon)\gamma\right) \FNorm{\matA - \matX_{\mathrm{opt}}
\matX_{\mathrm{opt}}\transp  \matA}^2.
$$
\end{theorem}
In words, given any set of points in some $n$-dimensional space and the number of clusters $k$, it suffices to select roughly $O(k \log k)$ actual features from the given points
and then run some $k$-means algorithm on this subset of the input. The theorem formally argues that the clustering it would be obtained in the low-dimensional space will be close to the clustering it would have been obtained after running the $k$-means method in the original high-dimensional data. 
We also state the result of the theorem in the notation we introduced in Section~\ref{sec:intro},
$$ \cl{F} (\mathcal{P}, \cl{S}_{\tilde{\gamma}}) \leq  \left(1+(2+\varepsilon)\gamma\right) \cl{F} (\mathcal{P}, \cl{S}_{opt}) .$$
Here, $ \cl{S}_{\tilde{\gamma}}$ is the partition obtained after running the $\gamma$-approximation $k$-means algorithm on the low-dimensional space. The approximation
factor is $\left(1+(2+\varepsilon)\gamma\right)$. The term $\gamma> 1$ is due to the fact that the $k$-means method that we run in the low-dimensional space does not recover the optimal $k$-means partition. The other factor $2+\varepsilon$ is due to the fact that we run $k$-means in the low-dimensional space.

\begin{proof} (of Theorem~\ref{fastkmeans})
We start by manipulating the term
$\FNormS{\matA -\matX_{\tilde{\gamma}} \matX_{\tilde{\gamma}}\transp \matA}$.
Notice that $\matA = \matA\matZ \matZ\transp + \matE$ (from Lemma~\ref{tropp2}). Also,
$$  \left( \left( \matI_{m } - \matX_{\tilde{\gamma}}\matX_{\tilde{\gamma}}\transp\right) \matA\matZ\matZ\transp \right)
\left( \left( \matI_{m } - \matX_{\tilde{\gamma}}\matX_{\tilde{\gamma}}\transp\right) \matE \right)\transp
= {\bf 0}_{m \times m},$$ because $\matZ\transp \matE\transp = {\bf 0}_{k \times m}$, by construction.
Now, using Matrix Pythagoras (see Lemma~\ref{lem:pythagoras}),
\begin{eqnarray}
\label{eqn:f1tt} \FNormS{\matA - \matX_{\tilde{\gamma}} \matX_{\tilde{\gamma}}\transp
\matA} = \underbrace{\FNormS{(\matI_{m } - \matX_{\tilde{\gamma}}
\matX_{\tilde{\gamma}}\transp) \matA\matZ \matZ\transp}}_{\theta_1^2} + \underbrace{\FNormS{(\matI_{m }
- \matX_{\tilde{\gamma}} \matX_{\tilde{\gamma}}\transp) \matE }}_{\theta_2^2}.
\end{eqnarray}
We first bound the second term of Eqn.~\eqref{eqn:f1tt}. Since
$\matI_{m }-\matX_{\tilde{\gamma}}\matX_{\tilde{\gamma}}\transp$ is a projection matrix,
it can be dropped without increasing the Frobenius norm (see Section~\ref{sec2}). Applying
Markov's inequality
on the equation of Lemma~\ref{tropp2}, we obtain that w.p. $0.99$, 
\begin{equation}\label{proof:short}
\FNormS{\matE} \le (1+100\varepsilon) \FNormS{\matA -\matA_k}.
\end{equation}
(See the Appendix for a short proof of this statement.)
Note also that $\matX_{\mathrm{opt}}\matX_{\mathrm{opt}}\transp\matA$ has rank at most $k$;
so, from the optimality of the SVD, overall,
$$
\theta_2^2 \le (1+100\varepsilon)\FNormS{\matA - \matA_k} \le (1+100\varepsilon)\FNormS{\matA - \matX_{\mathrm{opt}}\matX_{\mathrm{opt}}\transp\matA}=   (1+100\varepsilon) \mathrm{F}_{\mathrm{opt}}.
$$
We now bound the first term in Eqn.~\eqref{eqn:f1tt},
\begin{eqnarray}
\label{t0cc} \theta_1
&\leq&  \FNorm{(\matI_{m } -\matX_{\tilde{\gamma}}\matX_{\tilde{\gamma}}\transp)
\matA \matOmega \matS (\matZ\transp \matOmega \matS)^\dagger\matZ\transp} + \FNorm{\widetilde{\matE}} \\
\label{t1cc}
&\leq&  \FNorm{(\matI_{m } - \matX_{\tilde{\gamma}}\matX_{\tilde{\gamma}}\transp)\matA \matOmega \matS}
\TNorm{(\matZ\transp \matOmega \matS)^\dagger} + \FNorm{\widetilde{\matE}} \\
\label{t2cc}
&\leq& \sqrt{\gamma} \FNorm{(\matI_{m } - \matX_{\mathrm{opt}}\matX_{\mathrm{opt}}\transp)\matA \matOmega \matS}
\TNorm{(\matZ\transp \matOmega \matS)^\dagger}  + \FNorm{\widetilde{\matE}}
\end{eqnarray}
In Eqn.~\eqref{t0cc}, we used Lemma~\ref{lem:rsall} (for an unspecified failure probability $\delta$;
also, $\tilde{\matE} \in \R^{m \times n}$ is from that lemma), the triangle
inequality, and the fact that $\matI_m -
\matX_{\tilde{\gamma}}\matX_{\tilde{\gamma}}\transp$ is a projection matrix and
can be dropped without increasing the Frobenius norm. In
Eqn.~\eqref{t1cc}, we used spectral submultiplicativity and the fact that $\matZ\transp$
can be dropped without changing the spectral norm. In Eqn.~\eqref{t2cc}, we replaced
$\matX_{\tilde{\gamma}}$ by $\matX_{\mathrm{opt}}$ and the factor $\sqrt{\gamma}$ appeared
in the first term. To better understand this step, notice that
$\matX_{\tilde{\gamma}}$ gives a $\gamma$-approximation to the optimal
$k$-means clustering of $\matC = \matA \matOmega \matS$, so any other $m \times
k$ indicator matrix (e.g. $\matX_{\mathrm{opt}}$) satisfies,
$$\FNormS{\left(\matI_{m } - \matX_{\tilde{\gamma}} \matX_{\tilde{\gamma}}\transp\right) \matA \matOmega \matS}
\leq \gamma \min_{\matX \in \cal{X}} \FNormS{(\matI_{m } - \matX \matX\transp) \matA \matOmega \matS} \leq
\gamma \FNormS{\left(\matI_{m} - \matX_{\mathrm{opt}} \matX_{\mathrm{opt}}\transp\right) \matA \matOmega \matS}.$$
By using Lemma \ref{lem:fnorm} with $\delta = 3/4$ and Lemma~\ref{lem:random} (for an unspecified failure probability $\delta$),
$$ \FNorm{(\matI_{m} -\matX_{\mathrm{opt}}\matX_{\mathrm{opt}}\transp)\matA \matOmega \matS}
\TNorm{(\matZ\transp \matOmega \matS)^\dagger} \leq  \sqrt{ \frac{4}{3-3\varepsilon} \mathrm{F}_{\mathrm{opt}}}.$$
We are now in position to bound $\theta_1$.  In Lemmas~\ref{lem:rsall} and~\ref{lem:random}, let $\delta=0.01$.
Assuming $1 \leq \gamma$,
$$ \theta_1 \leq \left(  \sqrt{\frac{4}{3-3\varepsilon}} + \frac{1.6 \varepsilon \sqrt{1+100\varepsilon}}{\sqrt{0.01}}\right) \sqrt{\gamma } \sqrt{\mathrm{F}_{\mathrm{opt}}}
\leq \left(  \sqrt{2} + 94 \varepsilon  \right) \sqrt{\gamma } \sqrt{\mathrm{F}_{\mathrm{opt}} }.
$$
The last inequality follows from our choice of $\varepsilon < 1/3$ and elementary algebra. Taking squares on both sides,
$$ \theta_{1}^2 \le \left(  \sqrt{2} + 94 \varepsilon  \right)^2  \gamma \mathrm{F}_{\mathrm{opt}}
 \le (2 + 3900 \varepsilon) \gamma \mathrm{F}_{\mathrm{opt}}.$$
Overall (assuming $1 \le \gamma $),
$$ \FNormS{\matA -\matX_{\tilde{\gamma}} \matX_{\tilde{\gamma}}\transp \matA} \le \theta_{1}^2 + \theta_{2}^2 \le
(2 + 3900 \varepsilon) \gamma \mathrm{F}_{\mathrm{opt}} + (1+100\varepsilon) \mathrm{F}_{\mathrm{opt}} \le \mathrm{F}_{\mathrm{opt}} + ( 2 + 4\cdot10^3 \varepsilon  )\gamma \mathrm{F}_{\mathrm{opt}}.
$$
Rescaling $\varepsilon$ accordingly ($c_1= 16\cdot10^6$) gives the bound in the Theorem. The failure probability
follows by a union bound on Lemma \ref{lem:fnorm} (with $\delta = 3/4$), Lemma~\ref{lem:rsall} (with $\delta=0.01$),
Lemma~\ref{lem:random} (with $\delta=0.01$), Lemma~\ref{tropp2} (followed by Markov's inequality with $\delta=0.01$),
and Definition~\ref{def:approx} (with failure probability $\delta_{\gamma}$). Indeed, $0.75 + 3 \cdot 0.01 + 0.01 + 0.01 + \delta_{\gamma}= 0.8 + \delta_{\gamma}$
is the overall failure probability, hence the bound in the theorem holds w.p. $0.2 - \delta_{\gamma}$.
\end{proof}

\section{Feature Extraction with Random Projections}\label{sec5}

\begin{algorithm}[t]
\begin{framed}
\textbf{Input:} Dataset $\matA\in\R^{m\times n}$, number of clusters $k$, and \math{0 < \varepsilon < \frac{1}{3}}. \\
\noindent \textbf{Output:} \math{\matC \in \R^{m \times r}} with $r = O(k / \varepsilon^2)$ artificial features.
\begin{algorithmic}[1]
\STATE Set  $r =  c_2 \cdot k /\varepsilon^2$, for a sufficiently large constant $c_2$ (see proof).
\STATE Compute a random $n \times r$ matrix $\matR$ as follows. For all $i=1,\dots ,n$, $j=1, \dots ,r$ (i.i.d.)
   \[ \matR_{ij} = \begin{cases}
       +1/\sqrt{r}, \text{w.p. 1/2},\\
      -1/\sqrt{r}, \text{w.p. 1/2}.
\end{cases} \]
\STATE Compute $\matC = \matA \matR $ with the Mailman Algorithm (see text).
\STATE Return $\matC\in \R^{m \times r}$.
\end{algorithmic}
\caption{Randomized Feature Extraction for $k$-means Clustering.}
\label{alg:chap62}
\end{framed}
\end{algorithm}
We prove that any set of $m$ points in $n$
dimensions (rows in a matrix $\matA \in \R^{m \times n}$) can be projected into
$r = O(k / \varepsilon^2)$ dimensions in
$O(m n \lceil \varepsilon^{-2} k/ \log(n) \rceil )$ time such that, with
constant probability, the objective value of the optimal $k$-partition of the points
is preserved within a factor of $2+\varepsilon$. The projection is
done by post-multiplying $\matA$ with an $n \times r$ random
matrix $\matR$ having entries $+1/\sqrt{r}$ or $-1/\sqrt{r}$ with equal probability.
%

%
The algorithm needs $O(m k /\varepsilon^2)$ time to generate $\matR$; then,
the product $\matA \matR$ can be naively computed in $O(mnk/\varepsilon^2)$. However, one can employ the so-called mailman algorithm for matrix
multiplication~\cite{LZ09} and compute the product $\matA \matR$ in $O(m n \lceil \varepsilon^{-2} k/ \log(n) \rceil )$.
Indeed, the mailman algorithm computes (after preprocessing) a matrix-vector product of any $n$-dimensional vector (row of $\matA$) with an
$n \times \log(n)$ sign matrix in $O(n)$ time.
Reading the input $n \times \log n$ sign matrix requires $O(n\log n)$ time. However, in our case we only consider
multiplication with a random sign matrix, therefore we can avoid the preprocessing step by directly computing
a random correspondence matrix as discussed in~\cite[Preprocessing Section]{LZ09}.
By partitioning the columns of our $n \times r$ matrix $\matR$ into
$\lceil r/\log(n)\rceil$ blocks, the desired running time follows.

Theorem \ref{thm:second_result} is our quality-of-approximation result
regarding the clustering that can be obtained with the features returned
from Algorithm~\ref{alg:chap62} .
Notice that if $\gamma = 1$, the distortion is at most $2+\varepsilon$, as advertised in Table~\ref{tb:sum}.
If the $\gamma$-approximation algorithm is~\cite{KSS04} the overall approximation factor
would be $(1 + (1+\varepsilon)^2) = 2 + O(\varepsilon)$ with running time of the order
$O( m n \lceil \varepsilon^{-2} k/ \log(n) \rceil + 2^{(k/\varepsilon)^{O(1)}} m k / \varepsilon^2 )$.
\begin{theorem}\label{thm:second_result}
Let $\matA \in \R^{m \times n}$ and $k$ be the inputs of the $k$-means clustering problem.
Let $\varepsilon \in (0,1/3)$ and construct features $\matC \in \R^{m \times r}$ with $r=O(k/\varepsilon^2)$ by
using Algorithm~\ref{alg:chap62} in $O(m n \lceil \varepsilon^{-2} k/ \log(n) \rceil )$ time.
Run any $\gamma$-approximation $k$-means algorithm with failure probability $\delta_{\gamma}$
on $\matC, k$ and construct $\matX_{\tilde{\gamma}}$.
Then, w.p. at least $0.96 - \delta_{\gamma}$,
$$
\FNorm{\matA - \matX_{\tilde{\gamma}} \matX_{\tilde{\gamma}}\transp  \matA}^2
\leq \left(1+(1+\varepsilon)\gamma\right) \FNorm{\matA - \matX_{\mathrm{opt}}
\matX_{\mathrm{opt}}\transp  \matA}^2.
$$
\end{theorem}
In words, given any set of points in some $n$-dimensional space and the number of clusters $k$, it suffices to create (via random projections) roughly $O(k)$ new features 
and then run some $k$-means algorithm on this new input. The theorem formally argues that the clustering it would be obtained in the low-dimensional space will be close to the clustering it would have been obtained after running the $k$-means method in the original high-dimensional data. 
We also state the result of the theorem in the notation we introduced in Section~\ref{sec:intro},
$$ \cl{F} (\mathcal{P}, \cl{S}_{\tilde{\gamma}}) \leq  \left(1+(1+\varepsilon)\gamma\right) \cl{F} (\mathcal{P}, \cl{S}_{opt}) .$$
Here, $ \cl{S}_{\tilde{\gamma}}$ is the partition obtained after running the $\gamma$-approximation $k$-means algorithm on the low-dimensional space. 
The approximation
factor is $\left(1+(1+\varepsilon)\gamma\right)$. The term $\gamma> 1$ is due to the fact that the $k$-means method that we run in the low-dimensional space does not recover the optimal $k$-means partition. The other factor $1+\varepsilon$ is due to the fact that we run $k$-means in the low-dimensional space. 

\begin{proof}(of Theorem~\ref{thm:second_result})
We start by manipulating the term
$\FNormS{\matA -\matX_{\tilde{\gamma}} \matX_{\tilde{\gamma}}\transp \matA}$.
Notice that $\matA = \matA_k + \matA_{\rho-k}$. Also, $
\left( \left( \matI_{m } - \matX_{\tilde{\gamma}}\matX_{\tilde{\gamma}}\transp\right) \matA_k \right)
\left( \left( \matI_{m } - \matX_{\tilde{\gamma}}\matX_{\tilde{\gamma}}\transp\right) \matA_{\rho-k} \right)\transp
= {\bf 0}_{m \times m}$, because $\matA_k \matA_{\rho-k}\transp = {\bf 0}_{m \times m}$, by the orthogonality of
the corresponding subspaces. Now, using Lemma~\ref{lem:pythagoras},
\begin{eqnarray}\label{eqn:f1}
\FNorm{\matA - \matX_{\tilde{\gamma}} \matX_{\tilde{\gamma}}\transp \matA}^2\ =\ \underbrace{\FNorm{(\matI_m -
\matX_{\tilde{\gamma}} \matX_{\tilde{\gamma}}\transp ) \matA_k}^2}_{\theta_3^2}\ +\ \underbrace{\FNorm{(\matI_m - \matX_{\tilde{\gamma}}
\matX_{\tilde{\gamma}}\transp )\matA_{\rho-k}}^2}_{\theta_4^2}.
\end{eqnarray}
We first bound the second term of Eqn.~\eqref{eqn:f1}. Since
$\matI_m-\matX_{\tilde{\gamma}}\matX_{\tilde{\gamma}}\transp $ is a projection
matrix, it can be dropped without increasing the Frobenius norm. So, by using
this and the fact that $\matX_{\mathrm{opt}}\matX_{\mathrm{opt}}\transp \matA$ has rank at most $k$,
\begin{eqnarray} \label{eqn:f2}
\theta_4^2\ \leq\ \FNorm{\matA_{\rho-k}}^2\  = \FNorm{\matA - \matA_k}^2\ \leq\ \FNorm{ \matA -
\matX_{\mathrm{opt}}\matX_{\mathrm{opt}}\transp \matA }^2.
\end{eqnarray}
We now bound the first term of Eqn.~\eqref{eqn:f1},
\begin{eqnarray}
\label{t0} \theta_3
&\leq&  \FNorm{(\matI_m -\matX_{\tilde{\gamma}}\matX_{\tilde{\gamma}}\transp )\matA \matR(\matV_k \matR)^\dagger\matV_k\transp } + \FNorm{\widetilde{\matE}} \\
\label{t1}
&\leq&  \FNorm{(\matI_m -\matX_{\tilde{\gamma}}\matX_{\tilde{\gamma}}\transp )\matA \matR} \TNorm{(\matV_k \matR)^\dagger} + \FNorm{\widetilde{\matE}}\\
\label{t2}
&\leq& \sqrt{\gamma} \FNorm{(\matI_m - \matX_{\mathrm{opt}}\matX_{\mathrm{opt}}\transp )\matA \matR} \TNorm{(\matV_k \matR)^\dagger}  + \FNorm{\widetilde{\matE}} \\
\label{t3}
&\leq& \sqrt{\gamma}  \sqrt{1+\varepsilon} \FNorm{(\matI_m - \matX_{\mathrm{opt}}\matX_{\mathrm{opt}}\transp )\matA}  \frac{1}{1-\varepsilon} + 3 \varepsilon \FNorm{\matA -\matA_{k}} \\
\label{t4}
&\leq& \sqrt{\gamma}  (1+2.5 \varepsilon) \FNorm{(\matI_m -\matX_{\mathrm{opt}}\matX_{\mathrm{opt}}\transp )\matA} +  3 \varepsilon\sqrt{\gamma} \FNorm{(\matI_m -\matX_{\mathrm{opt}}\matX_{\mathrm{opt}}\transp )\matA}\\
\label{t5}
&=& \sqrt{\gamma} ( 1 + 5.5 \varepsilon) \FNorm{(\matI_m
-\matX_{\mathrm{opt}}\matX_{\mathrm{opt}}\transp )\matA}
\end{eqnarray}
In Eqn.~\eqref{t0}, we used the second statement of Lemma \ref{lem:rpall}, the
triangle inequality for matrix norms, and the fact that $\matI_m -
\matX_{\tilde{\gamma}}\matX_{\tilde{\gamma}}\transp $ is a projection matrix
and can be dropped without increasing the Frobenius norm.
In Eqn.~\eqref{t1}, we used spectral submultiplicativity
and the fact that $\matV_k\transp $ can be dropped
without changing the spectral norm.
In Eqn.~\eqref{t2}, we replaced
$\matX_{\tilde{\gamma}}$ by $\matX_{\mathrm{opt}}$ and the factor $\sqrt{\gamma}$
appeared in the first term. To better understand this step, notice
that $\matX_{\tilde{\gamma}}$ gives a $\gamma$-approximation to the
optimal $k$-means clustering of the matrix $\matC$, and any other $m
\times k$ indicator matrix (for example, the matrix $\matX_{\mathrm{opt}}$)
satisfies,
\begin{equation*}
\FNorm{\left(\matI_m - \matX_{\tilde{\gamma}} \matX_{\tilde{\gamma}}\transp
\right) \matC}^2 \leq\ \gamma\ \min_{\matX \in \cal{X}} \FNorm{(\matI_m - \matX
\matX\transp ) \matC}^2\ \leq \gamma \FNorm{\left(\matI_m - \matX_{\mathrm{opt}}
\matX_{\mathrm{opt}}\transp \right) \matC}^2.
\end{equation*}
In Eqn.~\eqref{t3}, we used the first statement of
Lemma \ref{lem:rpall} and Lemma~\ref{lem:RP:FNorm} with
$ \matY= (\matI - \matX_{\mathrm{opt}} \matX_{\mathrm{opt}}\transp )\matA$.
In Eqn.~\eqref{t4}, we used the fact that
$\gamma \geq 1$, the optimality of SVD, and that for any $\varepsilon \in (0,1/3)$, $\sqrt{1+\varepsilon}/(1-\varepsilon) \leq 1+ 2.5\varepsilon$.
Taking squares in
Eqn.~\eqref{t5} we obtain,
\[ \theta_3^2\ \leq\ \gamma \left( 1 + 5.5 \varepsilon \right)^2 \FNorm{(\matI_m
-\matX_{\mathrm{opt}}\matX_{\mathrm{opt}}\transp )\matA}^2\
\leq\ \gamma ( 1 + 15 \varepsilon ) \FNorm{(\matI_m
-\matX_{\mathrm{opt}}\matX_{\mathrm{opt}}\transp )\matA}^2.\]
Rescaling $\varepsilon$ accordingly gives the approximation bound in the theorem ($c_2 = 3330 \cdot 15^2$).
The failure probability $0.04 + \delta_{\gamma}$ follows by a
union bound on the failure probability $\delta_{\gamma}$ of the $\gamma$-approximation $k$-means
algorithm (Definition~\ref{def:approx}), Lemma~\ref{lem:RP:FNorm}, and Lemma~\ref{lem:rpall}.
\end{proof}

\paragraph{Disscusion.} As we mentioned in Section~\ref{sec:prior}, one can project the data down to $O(\log (m) / \varepsilon^2)$ dimensions and guarantee a clustering error which is not more than $(1+\varepsilon)$ times the optimal clustering error. This result is straightforward using the Johnson-Lindenstrauss lemma, which asserts that after such a dimension reduction all pairwise (Euclidian) distances of the points would be preserved by a factor $(1+\varepsilon)$~\cite{JL84}. If distances are preserved, then all clusterings - hence the optimal one - are preserved by the same factor. 

Our result here extends the Johnson-Lindenstrauss result in a remarkable way. It argues that much less dimensions suffice to preserve the optimal clustering in the data. We do not prove that pairwise distances are preserved. Our proof uses the linear algebraic formulation of the $k$-means clustering problem and shows that if the spectral information of certain matrices is preserved then the $k$-means clustering is preserved as well. Our bound is worse than the relative error bound obtained with $O(\log (m) / \varepsilon^2)$ dimensions; we believe though that it is possible to obtain a relative error bound and the $(2+\varepsilon)$ bound might be an artifact of the analysis. 

\begin{algorithm}
\begin{framed}
\textbf{Input:} Dataset $\matA\in\R^{m\times n}$, number of clusters $k$, and \math{0 < \varepsilon < 1}. \\
\noindent \textbf{Output:} \math{\matC \in \R^{m \times k}} with $k$ artificial features.
\begin{algorithmic}[1]
\STATE Let $ \matZ = \text{FastFrobeniusSVD}(\matA, k, \varepsilon)$; $\matZ \in \R^{n \times k}$ (via Lemma~\ref{tropp2}).
\STATE Return $\matC = \matA \matZ \in \R^{m \times k}$.
\end{algorithmic}
\caption{Randomized Feature Extraction for $k$-means Clustering.}
\label{alg:chap63}
\end{framed}
\end{algorithm}
\section{Feature Extraction with Approximate SVD}\label{sec6}
Finally, we present a feature extraction algorithm that employs the SVD
to construct $r = k$ artificial features. Our method and proof techniques
are the same with those of~\cite{DFKVV99} with the only difference being the fact that we
use a fast approximate (randomized) SVD via Lemma~\ref{tropp2} as opposed
to the expensive exact deterministic SVD. In fact, replacing $\matZ = \matV_k$ reproduces
the proof in~\cite{DFKVV99}. Our choice gives a considerably
faster algorithm with approximation error comparable to the error in~\cite{DFKVV99}.
\begin{theorem}\label{thm:first_result}
Let $\matA \in \R^{m \times n}$ and $k$ be inputs of the $k$-means clustering problem.
Let $\varepsilon \in (0,1)$ and construct features $\matC \in \R^{m \times k}$
by using Algorithm~\ref{alg:chap63} in $O(m n k / \varepsilon )$ time.
Run any $\gamma$-approximation $k$-means algorithm with failure probability $\delta_{\gamma}$
on $\matC, k$ and construct $\matX_{\tilde{\gamma}}$.
Then, w.p. at least $0.99 - \delta_{\gamma}$,
$$
\FNorm{\matA - \matX_{\tilde{\gamma}} \matX_{\tilde{\gamma}}\transp  \matA}^2
\leq \left(1+(1+\varepsilon)\gamma\right) \FNorm{\matA - \matX_{\mathrm{opt}}
\matX_{\mathrm{opt}}\transp  \matA}^2.
$$
\end{theorem}
In words, given any set of points in some $n$-dimensional space and the number of clusters $k$, it suffices to create exactly $k$ new features 
(via an approximate Singular Value Decomposition)
and then run some $k$-means algorithm on this new dataset. The theorem formally argues that the clustering it would be obtained in the low-dimensional space will be close to the clustering it would have been obtained after running the $k$-means method in the original high-dimensional data. 
We also state the result of the theorem in the notation we introduced in Section~\ref{sec:intro}:
$$ \cl{F} (\mathcal{P}, \cl{S}_{\tilde{\gamma}}) \leq  \left(1+(1+\varepsilon)\gamma\right) \cl{F} (\mathcal{P}, \cl{S}_{opt}) .$$
Here, $ \cl{S}_{\tilde{\gamma}}$ is the partition obtained after running the $\gamma$-approximation $k$-means algorithm on the low-dimensional space. 
The approximation
factor is $\left(1+(1+\varepsilon)\gamma\right)$. The term $\gamma> 1$ is due to the fact that the $k$-means method that we run in the low-dimensional space does not recover the optimal $k$-means partition. The other factor $1+\varepsilon$ is due to the fact that we run $k$-means in the low-dimensional space. 
\begin{proof}(of Theorem~\ref{thm:first_result})
We start by manipulating the term
$\FNormS{\matA -\matX_{\tilde{\gamma}} \matX_{\tilde{\gamma}}\transp \matA}$.
Notice that $\matA = \matA\matZ \matZ\transp + \matE$. Also,
$
\left( \left( \matI_{m } - \matX_{\tilde{\gamma}}\matX_{\tilde{\gamma}}\transp\right) \matA\matZ\matZ\transp \right)
\left( \left( \matI_{m } - \matX_{\tilde{\gamma}}\matX_{\tilde{\gamma}}\transp\right) \matE \right)\transp
= {\bf 0}_{m \times m}$, because $\matZ\transp \matE\transp = {\bf 0}_{k \times m}$, by construction.
Now, using the Matrix Pythagorean theorem (see Lemma~\ref{lem:pythagoras} in Section~\ref{sec2}),
\begin{eqnarray}
\label{eqn:f1tt1} \FNormS{\matA - \matX_{\tilde{\gamma}} \matX_{\tilde{\gamma}}\transp
\matA} = \underbrace{\FNormS{(\matI_{m } - \matX_{\tilde{\gamma}}
\matX_{\tilde{\gamma}}\transp) \matA\matZ \matZ\transp}}_{\theta_1^2} + \underbrace{\FNormS{(\matI_{m }
- \matX_{\tilde{\gamma}} \matX_{\tilde{\gamma}}\transp) \matE }}_{\theta_2^2}.
\end{eqnarray}

We first bound the second term of Eqn.~\eqref{eqn:f1tt1}. Since
$\matI_{m }-\matX_{\tilde{\gamma}}\matX_{\tilde{\gamma}}\transp$ is a projection matrix,
it can be dropped without increasing the Frobenius norm (see Section~\ref{sec2}). Applying
Markov's inequality
on the equation of Lemma~\ref{tropp2}, we obtain that w.p. $0.99$, 
\begin{equation}\label{proof:short2}
\FNormS{\matE} \le (1+100\varepsilon) \FNormS{\matA -\matA_k}.
\end{equation}
(This is Eqn.~\ref{proof:short}, of which we provided a short proof in the Appendix.)
Note also that $\matX_{\mathrm{opt}}\matX_{\mathrm{opt}}\transp\matA$ has rank at most $k$;
so, from the optimality of the SVD, overall,
$$
\theta_2^2 \le (1+100\varepsilon)\FNormS{\matA - \matA_k} \le (1+100\varepsilon)\FNormS{\matA - \matX_{\mathrm{opt}}\matX_{\mathrm{opt}}\transp\matA}=   (1+100\varepsilon) \mathrm{F}_{\mathrm{opt}}.
$$
Hence, it follows that w.p. $0.99$,
$$
\theta_2^2 \le
(1+100\varepsilon) \mathrm{F}_{\mathrm{opt}}.
$$
We now bound the first term in Eqn.~\eqref{eqn:f1tt1},
\begin{eqnarray}
%
\label{cb2}
\label{cb1} \theta_1
&\leq&  \FNorm{(\matI_{m } - \matX_{\tilde{\gamma}}\matX_{\tilde{\gamma}}\transp)\matA \matZ}  \\
\label{cb3}
&\leq& \sqrt{\gamma} \FNorm{(\matI_{m } - \matX_{\mathrm{opt}}\matX_{\mathrm{opt}}\transp)\matA \matZ}\\
\label{cb4}
&\leq& \sqrt{\gamma} \FNorm{(\matI_{m } - \matX_{\mathrm{opt}}\matX_{\mathrm{opt}}\transp)\matA}
\end{eqnarray}
In Eqn.~\eqref{cb2}, we used spectral submultiplicativity and the fact that $\TNorm{\matZ\transp}=1$.
In Eqn.~\eqref{cb3},  we replaced
$\matX_{\tilde{\gamma}}$ by $\matX_{\mathrm{opt}}$ and the factor $\sqrt{\gamma}$ appeared
in the first term (similar argument as in the proof of Theorem~\ref{fastkmeans}).
In Eqn.~\eqref{cb4}, we used spectral submultiplicativity and the fact that $\TNorm{\matZ}=1$.
Overall (assuming $\gamma \ge 1$),
$$ \FNormS{\matA -\matX_{\tilde{\gamma}} \matX_{\tilde{\gamma}}\transp \matA} \le \theta_{1}^2 + \theta_{2}^2 \le
 \gamma \mathrm{F}_{\mathrm{opt}} + (1+100\varepsilon) \mathrm{F}_{\mathrm{opt}} \le \mathrm{F}_{\mathrm{opt}} + (1 + 10^2 \varepsilon  )\gamma \mathrm{F}_{\mathrm{opt}}.
$$
The failure probability is $0.01+\delta_{\gamma}$, from a union bound on Lemma~\ref{tropp2} and Definition~\ref{def:approx}.
Finally, rescaling $\varepsilon$ accordingly gives the approximation bound in the theorem.
\end{proof}

\section{Experiments}\label{sec:experiments}
This section describes a preliminary experimental evaluation of the feature selection and feature extraction algorithms presented in this work. We implemented the proposed algorithms in MATLAB~\cite{MATLAB:2011} and compared them against a few other prominent dimensionality reduction techniques such as the Laplacian scores~\cite{HCN06}. Laplacian scores is a popular feature selection method for clustering and classification. We performed all the experiments on a Mac machine with a dual core 2.8 Ghz processor and 8 GB of RAM. 

Our empirical findings are far from exhaustive, however they indicate that the feature selection and feature extraction algorithms presented in this work achieve a satisfactory empirical performance with rather small values of $r$ (far smaller than the theoretical bounds presented here). We believe that the large constants that appear in our theorems (see e.g., Theorem~\ref{thm:second_result}) are artifacts of our theoretical analysis and can be certainly improved. 
\subsection{Dimensionality Reduction Methods}
Given $m$ points described with respect to $n$ features and the number of clusters $k$, our goal is to select or construct $r$ features on which we execute Lloyd's algorithm for $k$-means on this constructed set of features. In this section, we experiment with various methods for selecting or constructing the features. The number of features to be selected or extracted is part of the input as well. In particular, in Algorithm~\ref{alg:featureSelect} we do not consider $\varepsilon$ to be part of the input. We test the performance of the proposed algorithms for various values of $r$, and we compare our algorithms against other feature selection and feature extraction methods from the literature, that we summarize below: 
%
%
\begin{description}
  \item[1. Randomized Sampling with Exact SVD (Sampl/SVD).] This corresponds to Algorithm~\ref{alg:featureSelect} with the following modification. In the first step of the algorithm, the matrix $\matZ$ is calculated to contain exactly the top $k$ right singular vectors of $\matA$. 
  \item[2. Randomized Sampling with Approximate SVD (Sampl/ApproxSVD).] This corresponds to Algorithm~\ref{alg:featureSelect} with $\varepsilon$ fixed to $1/3$. 
  \item[3. Random Projections (RP).] Here we use Algorithm~\ref{alg:chap62}. However, in our implementation we use the naive approach for the matrix-matrix multiplication in the third step (not the Mailman algorithm~\cite{LZ09}). 
  \item[4. SVD.] This is Algorithm~\ref{alg:chap63} with the following modification. In the first step of the algorithm, the matrix $\matZ$ is calculated to contain exactly the top $k$ right singular vectors of $\matA$.   
  \item[5. Approximate SVD (ApprSVD).] This corresponds to Algorithm~\ref{alg:chap63} with $\varepsilon$ fixed to $1/3$.
  \item[6. Laplacian Scores (LapScores).] This corresponds to the feature selection method described in~\cite{HCN06}. 
  We use the MATLAB code from \cite{LapScores} with the default parameters. In particular, in MATLAB notation we executed the following commands, 
\[\matW = \text{constructW}(\matA); \text{Scores} = \text{LaplacianScore}(\matA, \matW);\]
\end{description}

Finally, we also compare all these methods against evaluating the $k$-means algorithm in the full dimensional dataset which we denote by \texttt{kMeans}. 
\subsection{$k$-means method}
Although our theorems allow the use of any $\gamma$-approximation algorithm for $k$-means, in practice the Lloyd's algorithm performs very well~\cite{Llo82}. Hence, we employ the Lloyd's algorithm in our experiments. Namely, every time we mention \emph{``we run $k$-means''}, we mean that we run 500 iterations of the Lloyd's algorithm with 5 different random initializations and return the best outcome over all repetitions, i.e., in MATLAB notation we run the following command, \texttt{kmeans(A, k, `Replicates', 5, `Maxiter', 500)}. 
\subsection{Datasets}\label{subsec:data}
\begin{figure}[ht!]
\centering
\subfigure[\texttt{Synth} - Running time]{\includegraphics[width=0.475\textwidth]{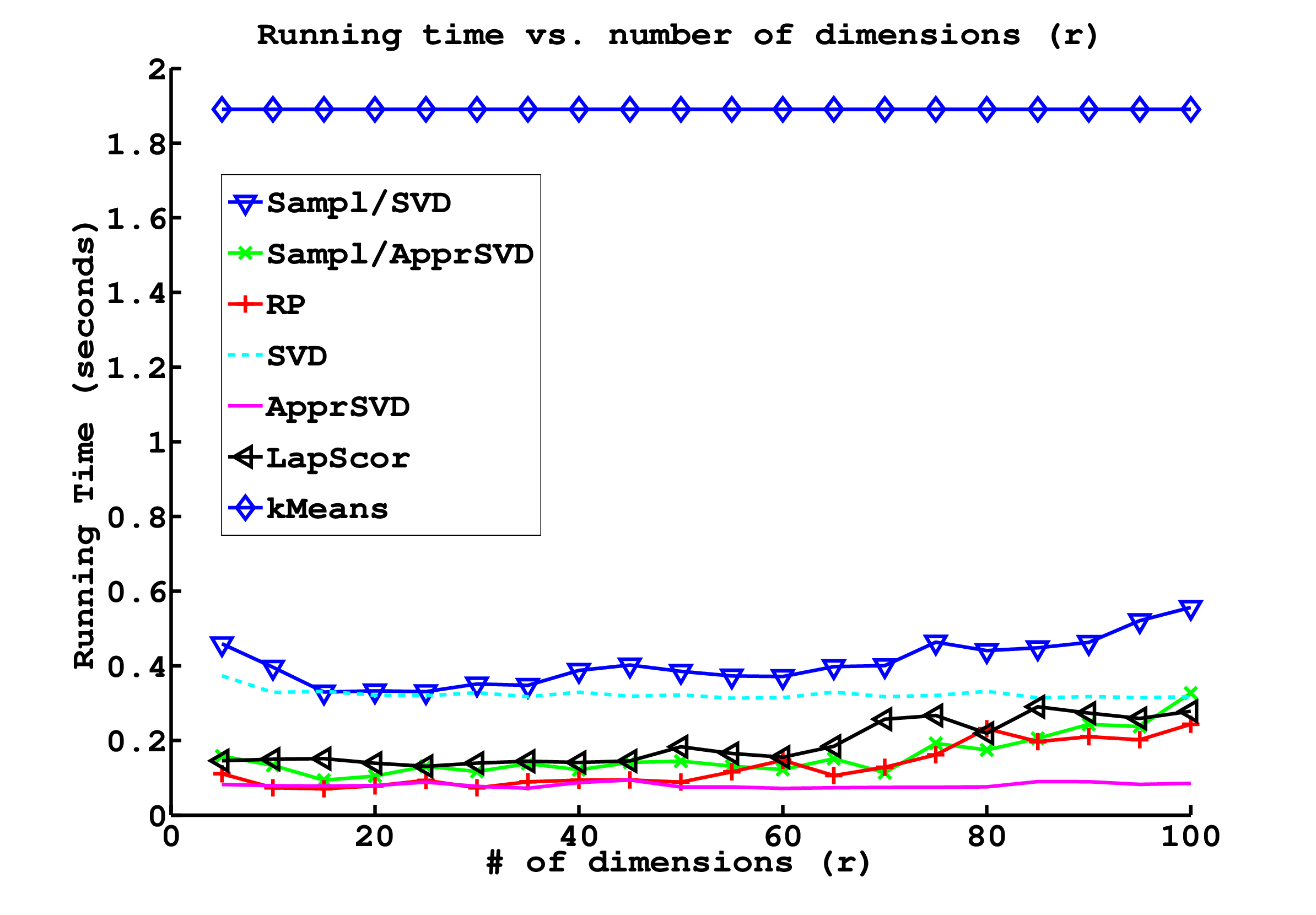}}
\subfigure[\texttt{USPS} - Running time]{\includegraphics[width=0.475\textwidth]{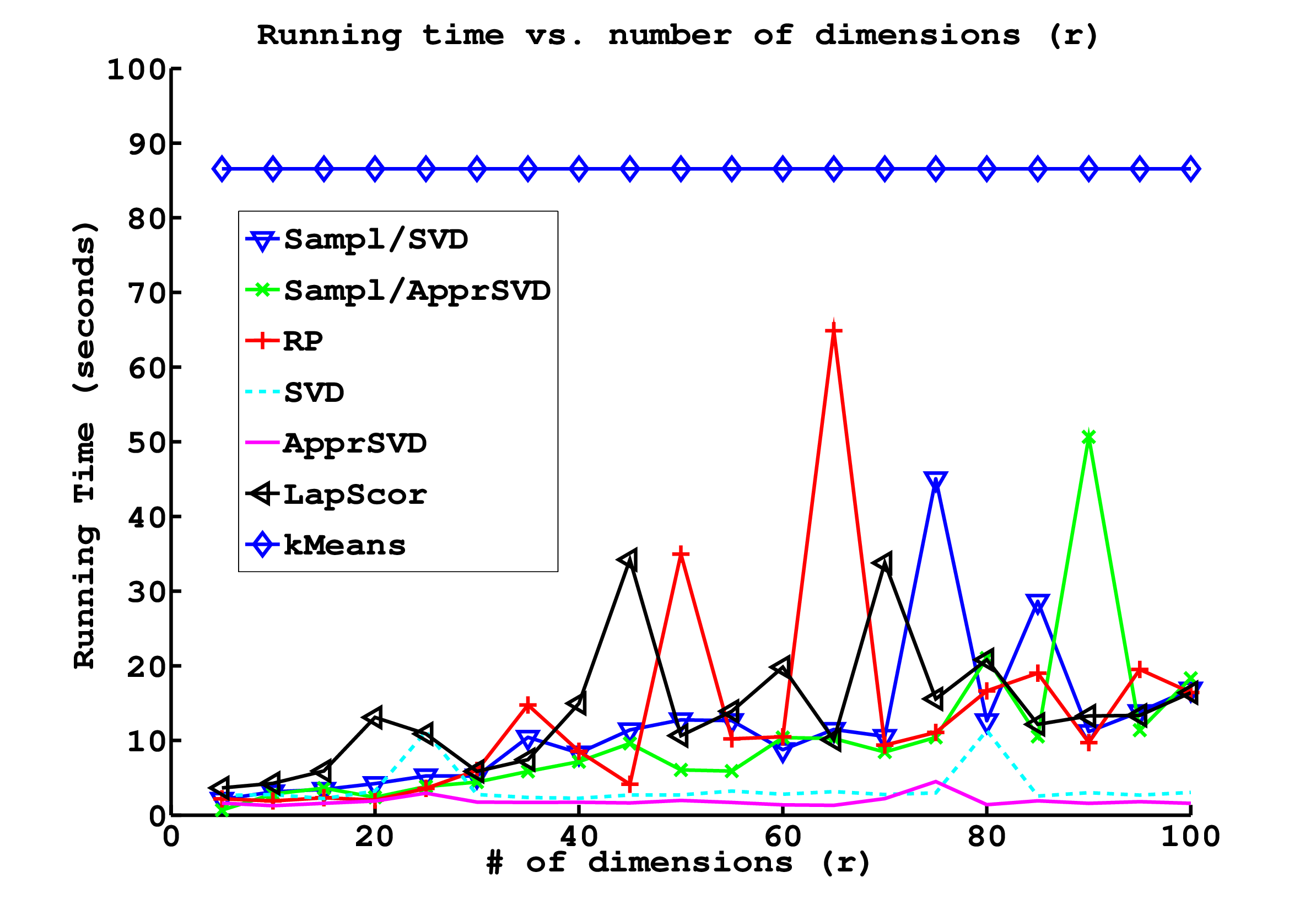}}\\
\subfigure[\texttt{Synth} - Objective]{\includegraphics[width=0.475\textwidth]{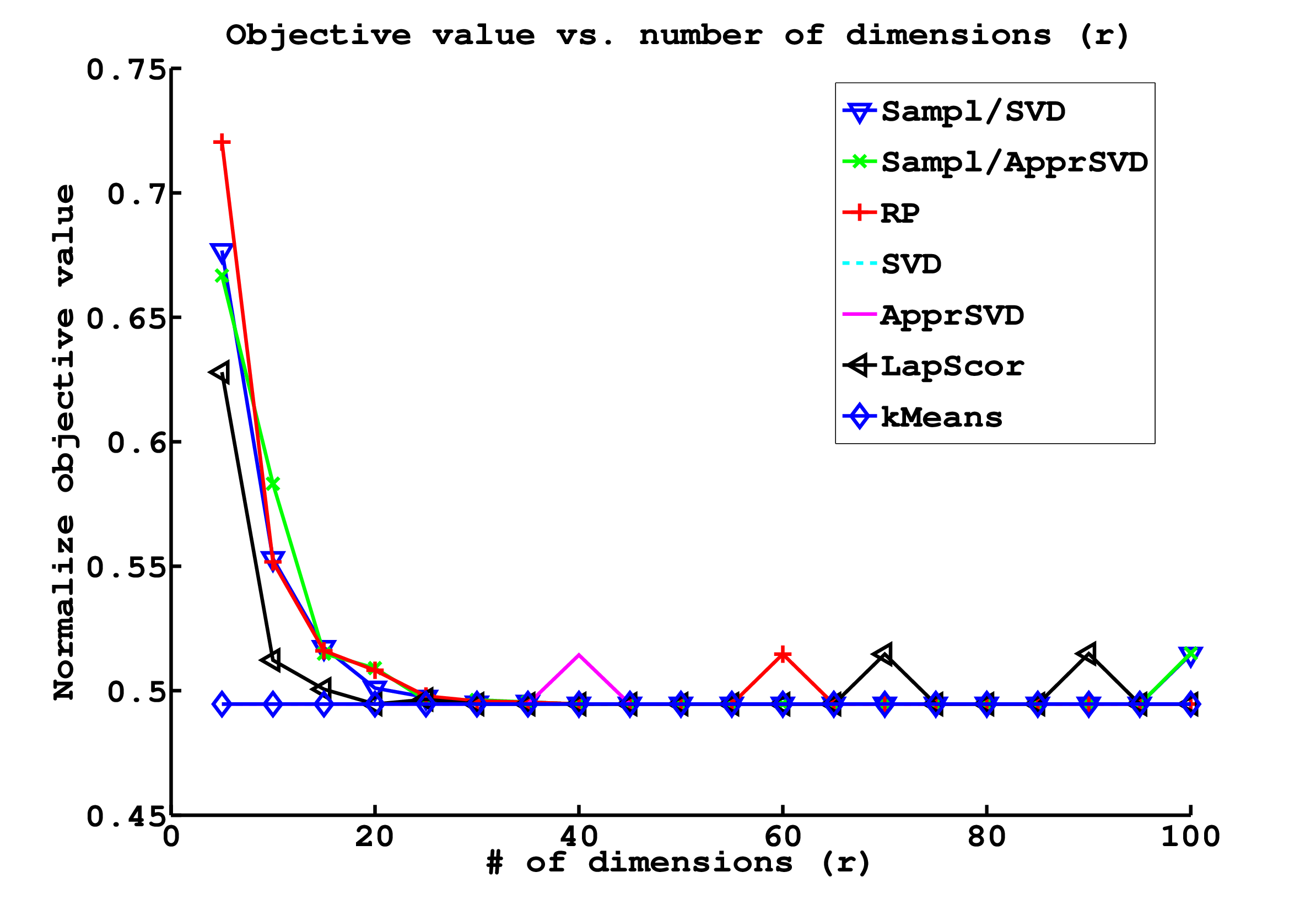}}
\subfigure[\texttt{USPS} - Objective]{\includegraphics[width=0.475\textwidth]{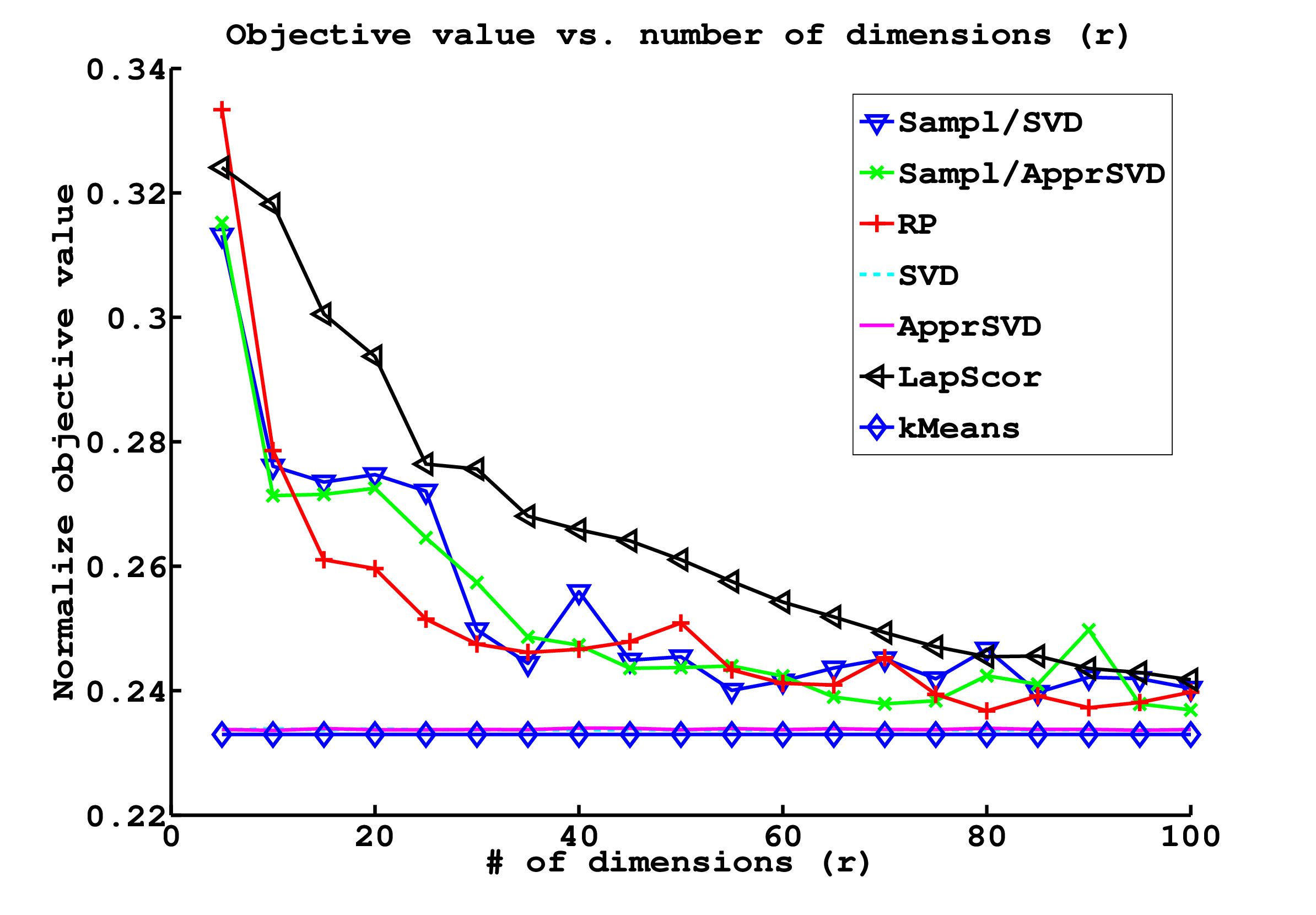}}\\
\subfigure[\texttt{Synth} - Accuracy of clustering]{\includegraphics[width=0.475\textwidth]{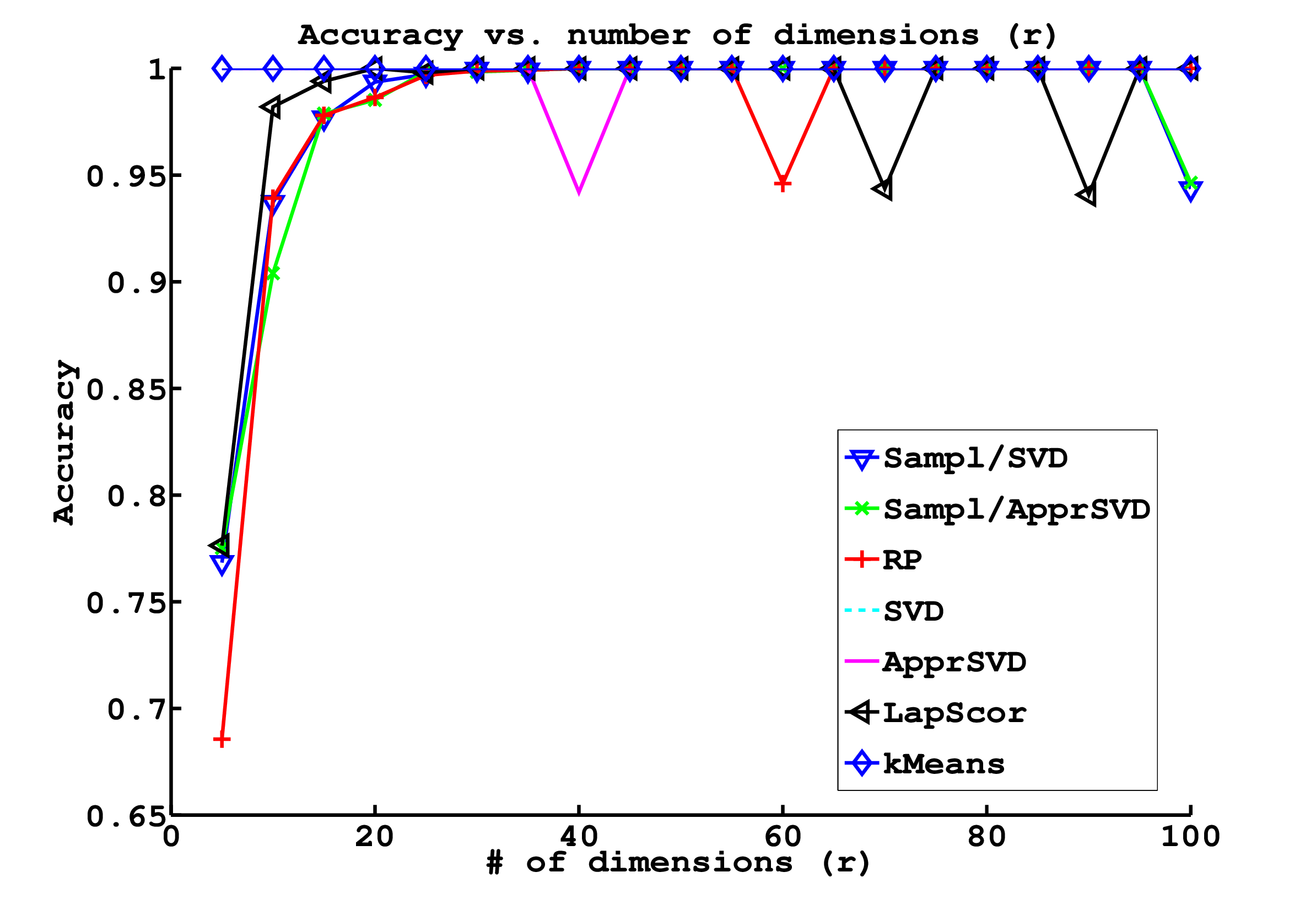}}
\subfigure[\texttt{USPS} - Accuracy of clustering]{\includegraphics[width=0.475\textwidth]{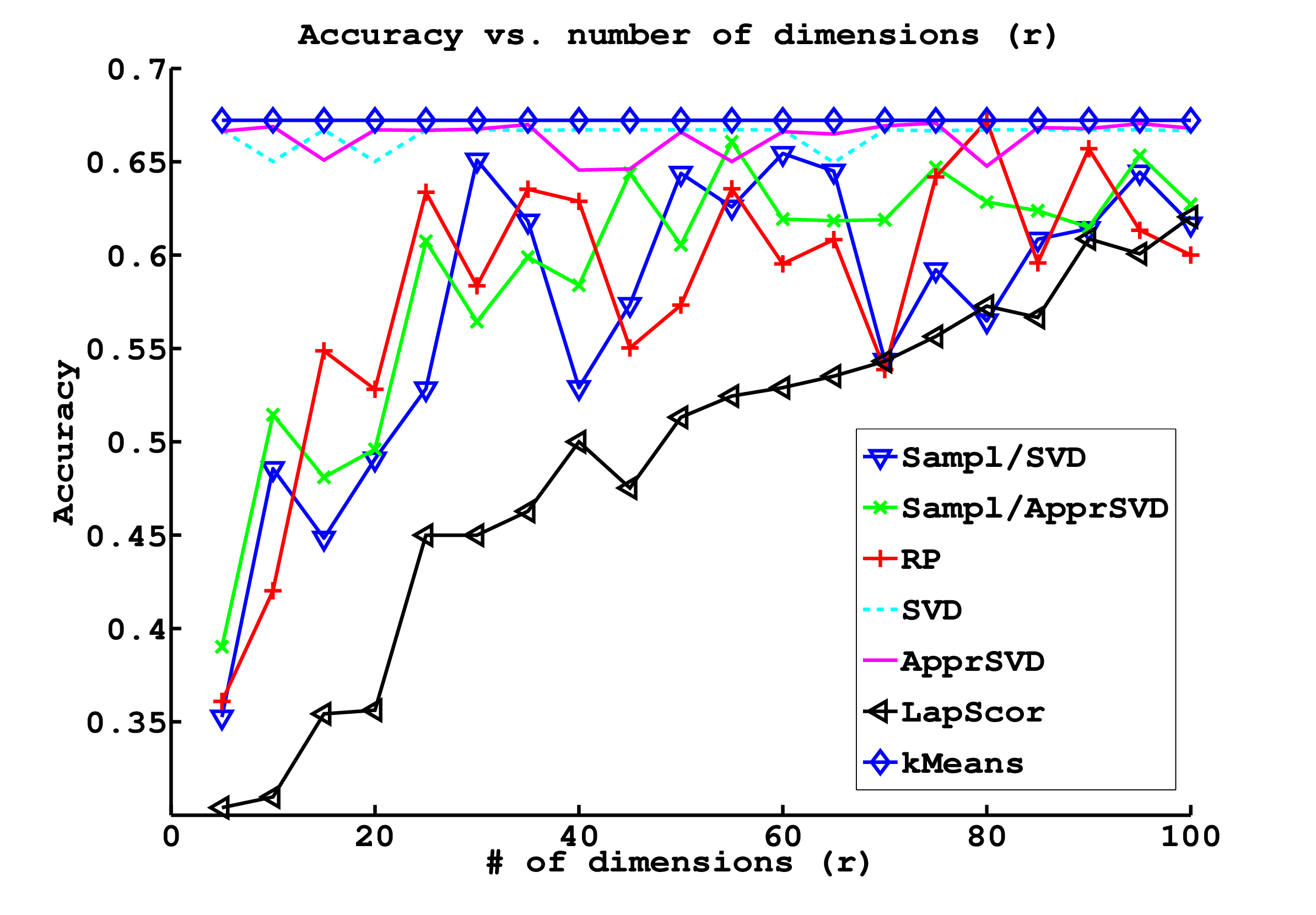}}
\caption{\small Plot of running time (a),(b), objective value (c),(d) and accuracy (e),(f) versus the number of projected dimensions for several dimensionality reduction approaches. Left column corresponds to the \texttt{Synth} dataset, whereas the right column corresponds to the \texttt{USPS} dataset.}
\label{fig:usps}
\end{figure}
\begin{figure}[ht!]
\centering
\subfigure[\texttt{COIL20} - Running time]{\includegraphics[width=0.475\textwidth]{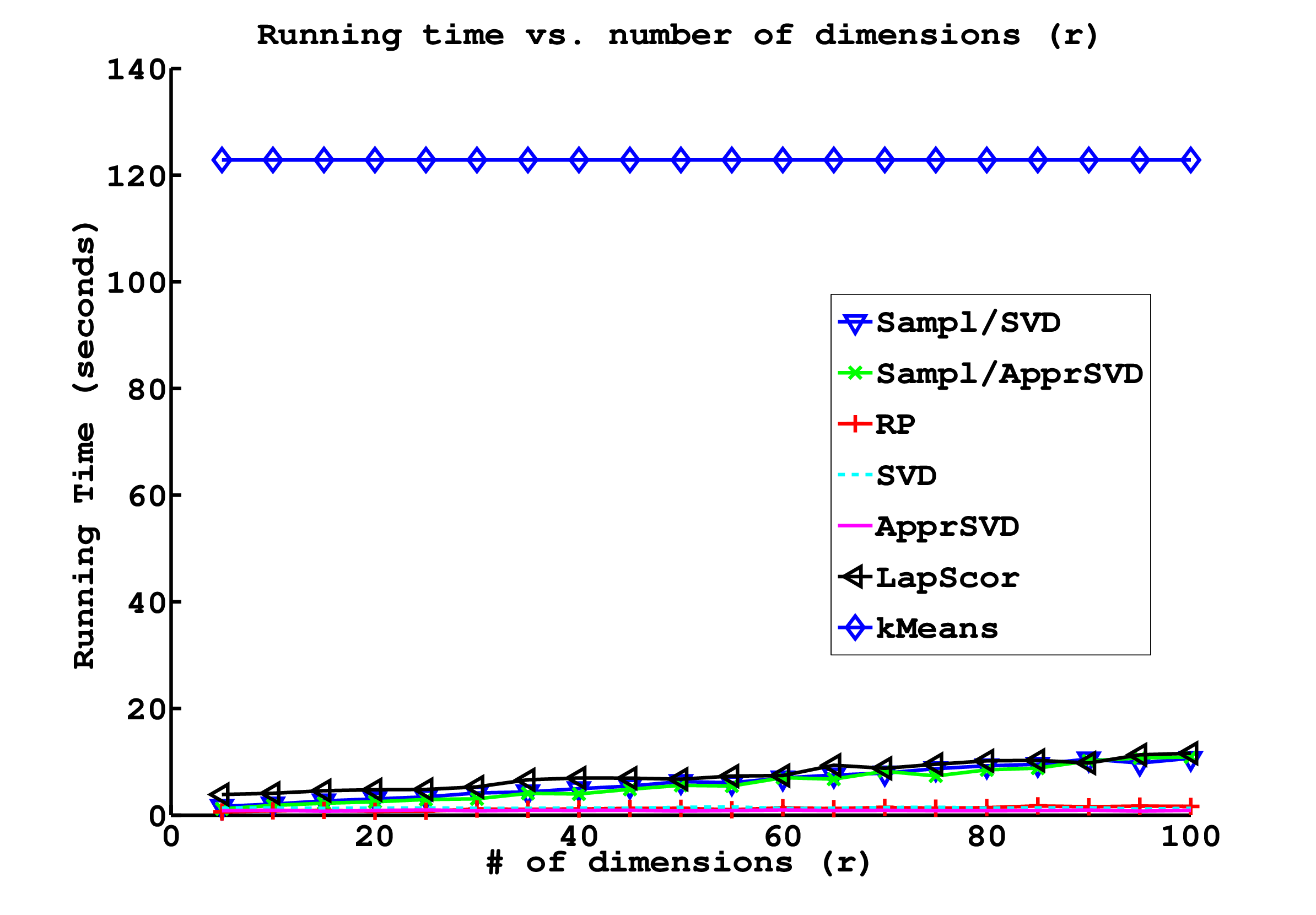}}
\subfigure[\texttt{LIGHT}  - Running time]{\includegraphics[width=0.475\textwidth]{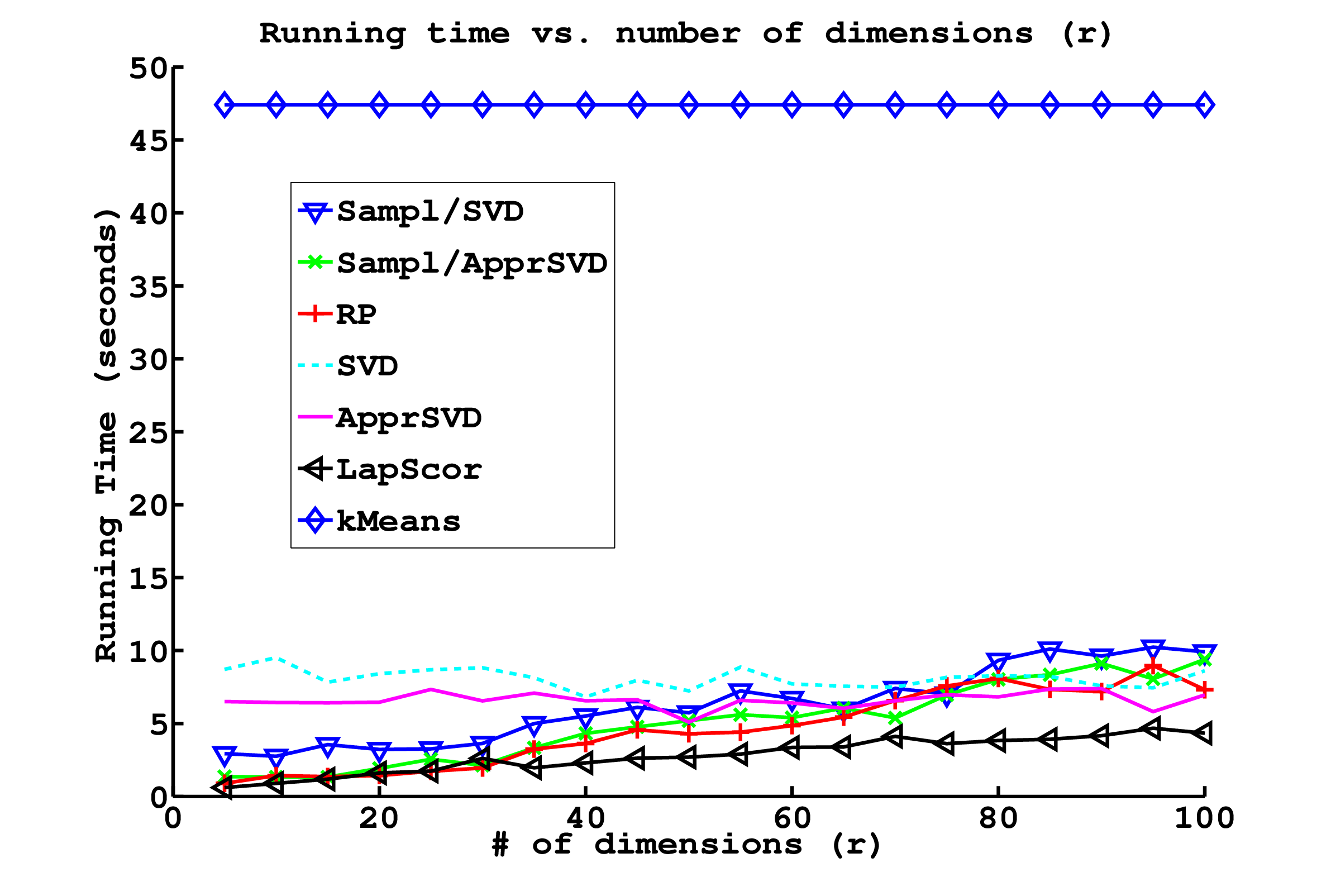}}\\
\subfigure[\texttt{COIL20} - Objective]{\includegraphics[width=0.475\textwidth]{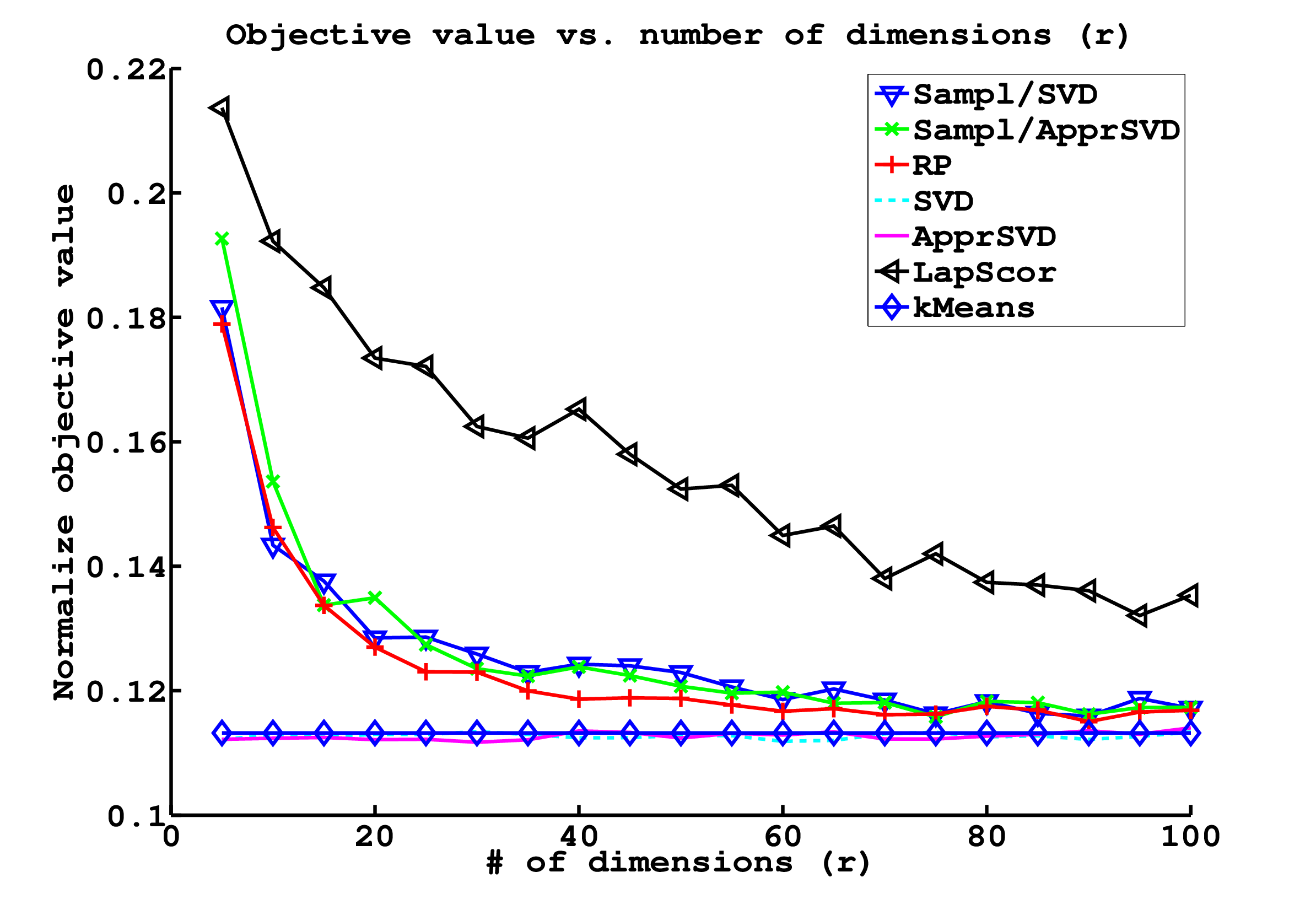}}
\subfigure[\texttt{LIGHT}  - Objective]{\includegraphics[width=0.475\textwidth]{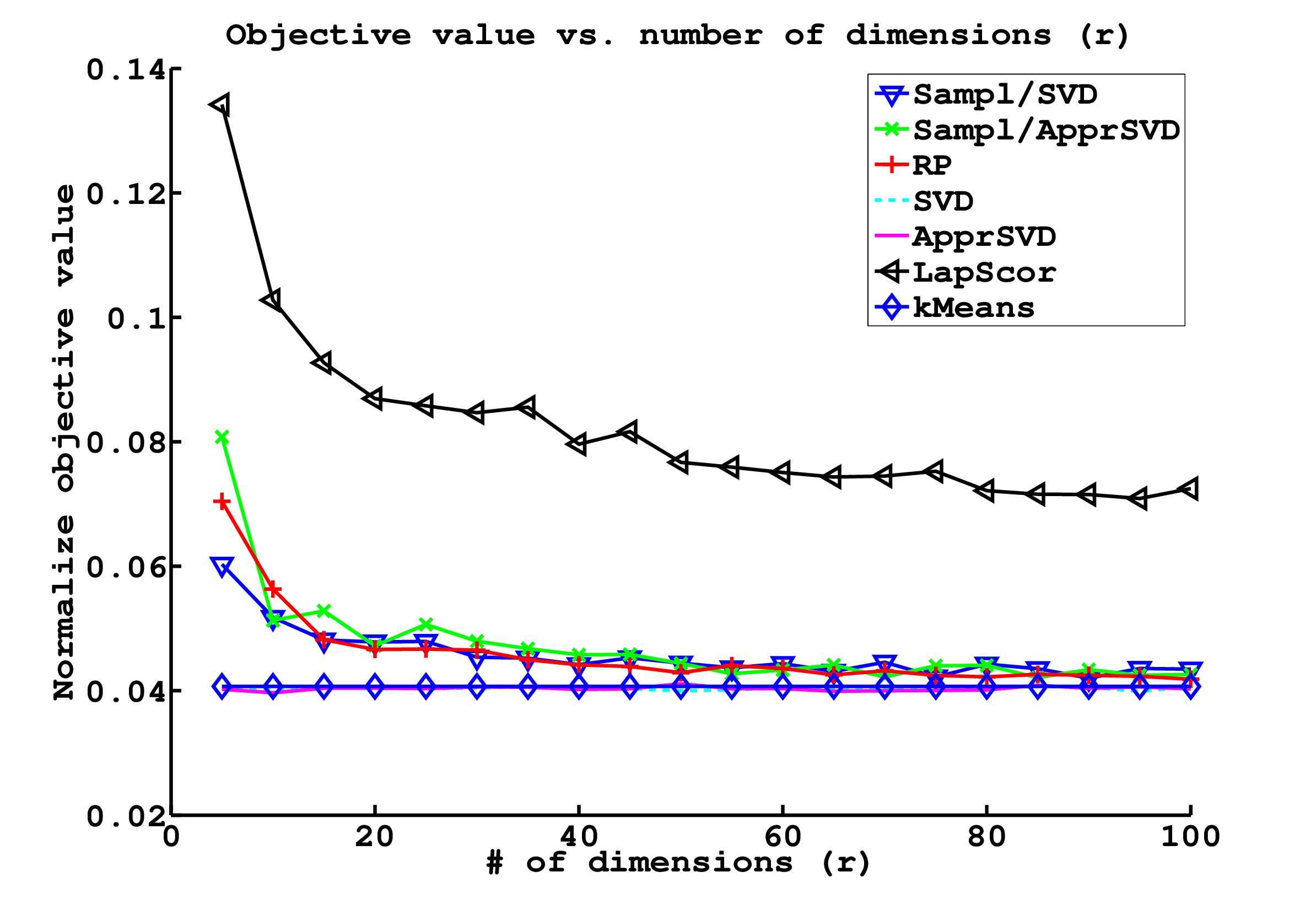}}\\
\subfigure[\texttt{COIL20} - Accuracy of clustering]{\includegraphics[width=0.475\textwidth]{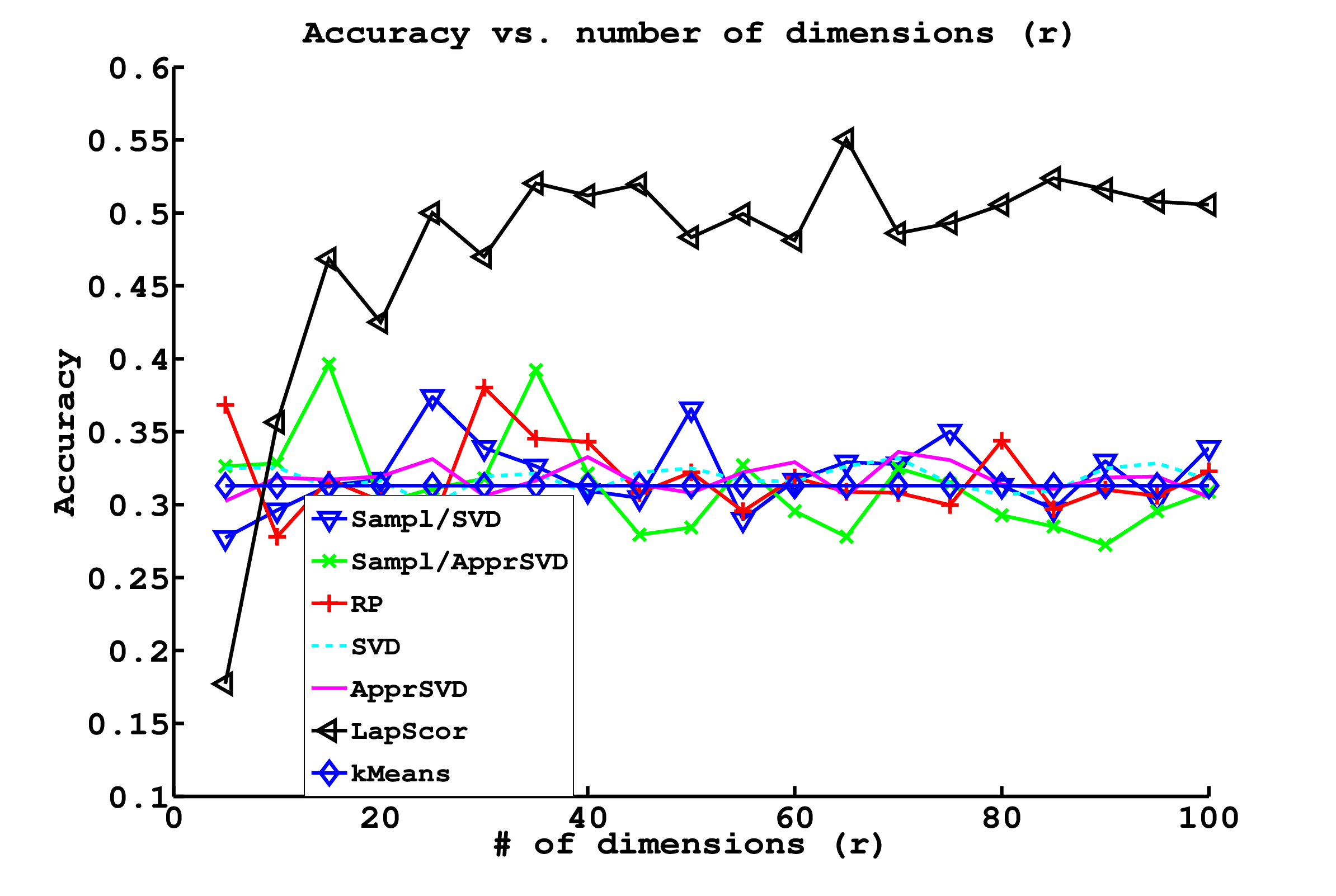}}
\subfigure[\texttt{LIGHT}  - Accuracy of clustering]{\includegraphics[width=0.475\textwidth]{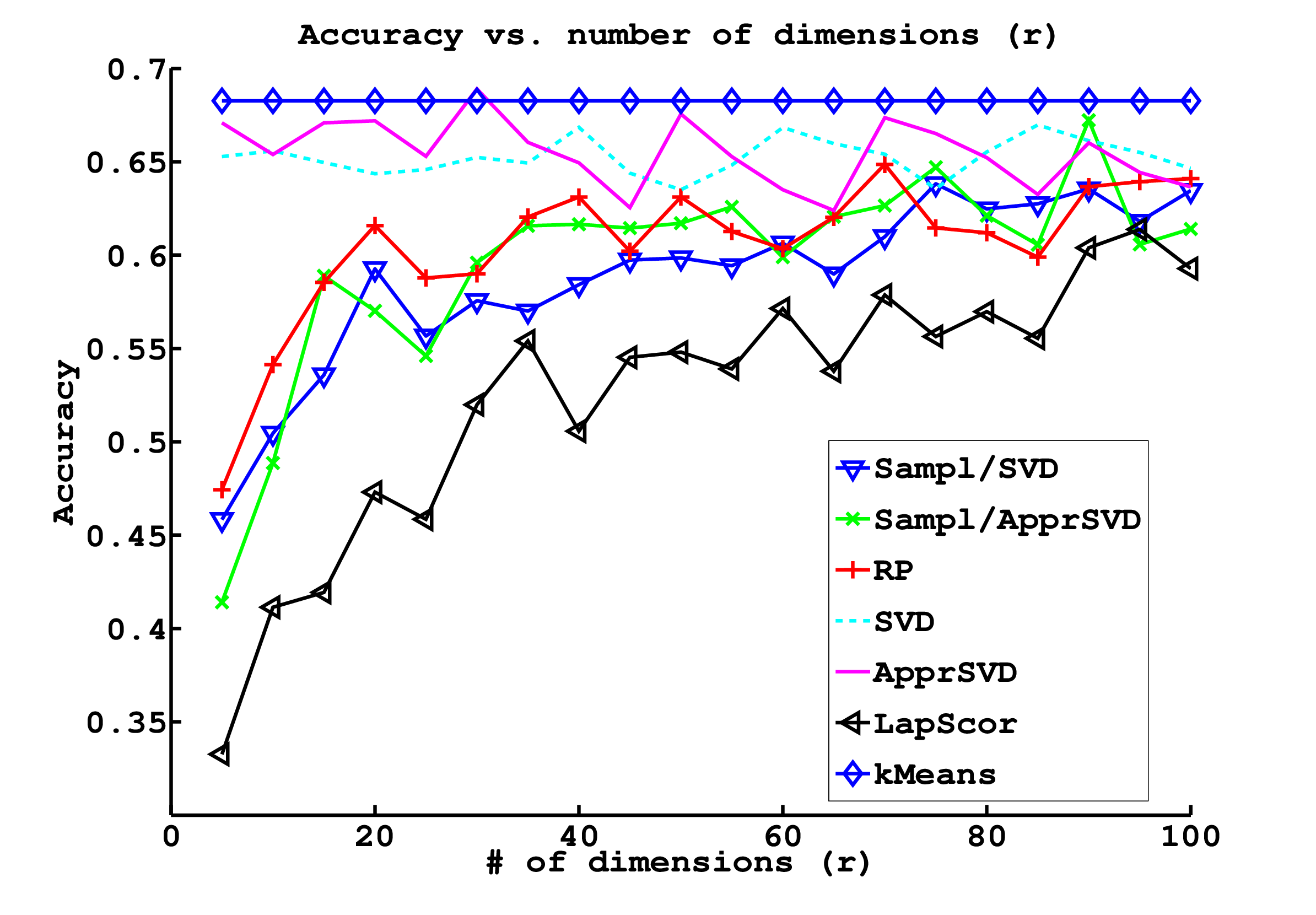}}\\
\caption{\small Plot of running time (a),(b), objective value (c),(d) and accuracy (e),(f) versus the number of projected dimensions for several dimensionality reduction approaches. Left column corresponds to the \texttt{COIL20} dataset, whereas the right column corresponds to the \texttt{LIGHT} dataset.}
\label{fig:light}
\end{figure}
We performed experiments on a few real-world and synthetic datasets. For the synthetic dataset, we generated a dataset of $m=1000$ points in $n=2000$ dimensions as follows. We chose $k=5$ centers uniformly at random from the $n$-dimensional hypercube of side length $2000$ as the ground truth centers. We then generated points from a Gaussian distribution of variance one, centered at each of the real centers. To each of the $5$ centers we generated $200$ points (we did not include the centers in the dataset). Thus, we obtain a number of well separated Gaussians with the real centers providing a good approximation to the optimal clustering. We will refer to this dataset as \texttt{Synth}.

For the real-world datasets we used five datasets that we denote by \texttt{USPS}, \texttt{COIL20}, \texttt{ORL}, \texttt{PIE} and \texttt{LIGHT}. The \texttt{USPS} digit dataset contains grayscale pictures of handwritten digits and can be downloaded from the UCI repository~\cite{UCI}. Each data point of \texttt{USPS} has 256 dimensions and there are 1100 data points per digit. The coefficients of the data points have been normalized between $0$ and $1$. The \texttt{COIL20} dataset contains 1400 images of 20 objects (the images of each objects were taken 5 degrees apart as the object is rotated on a turntable and each object has 72 images) and can be downloaded from~\cite{COIL20}. The size of each image is 32x32 pixels, with 256 grey levels per pixel. Thus, each image is represented by a 1024-dimensional vector. \texttt{ORL} contains ten different images each of 40 distinct subjects and can be located at~\cite{ORL}. For few subjects, the images were taken at different times, varying the lighting, facial expressions and facial details. All the images were taken against a dark homogeneous background with the subjects in an upright, frontal position (with tolerance for some side movement). There are in total 400 different objects having 4096 dimensions.

\texttt{PIE} is a database of 41,368 images of 68 people, each person under 13 different poses, 43 different illumination conditions, and with 4 different expressions~\cite{PIE}. Our dataset contains only five near frontal poses (C05, C07, C09, C27, C29) and all the images under different illuminations and expressions. Namely, there are in total $2856$ data points with $1024$ dimensions. The \texttt{LIGHT} dataset is identical with the dataset that has been used in ~\cite{HCN06}, the data points of \texttt{LIGHT} is $1428$ containing $1014$ features.
For each real-world dataset we fixed $k$ to be equal to the cardinality of their corresponding label set.
\subsection{Evaluation Methodology}
As a measure of quality of all methods we measure and report the objective function $\mathcal{F}$ of the $k$-means clustering problem. In particular, we
report a normalized version of $\mathcal{F}$, i.e. 
\[\mathcal{F} = \mathcal{F} / \| \matA \|_{\mathrm{F}}^2.\]
In addition, we report the mis-classification accuracy of the clustering result based on the labelled information of the input data. We denote this number by $P$ ($0 \leq P \leq 1$), where $P=0.9$, for example, implies that $90 \%$ of the points were assigned to the ``correct cluster''/label after the application of the clustering algorithm. Finally,  we report running times (in seconds). It is important to highlight that we report the running time of both the dimensionality reduction procedure and the $k$-means algorithm applied on the low-dimensional projected space for all proposed algorithms. All the reported quantities correspond to the average values of five independent executions.

\begin{figure}[ht!]
\centering
\subfigure[\texttt{PIE} - Running time]{\includegraphics[width=0.45\textwidth]{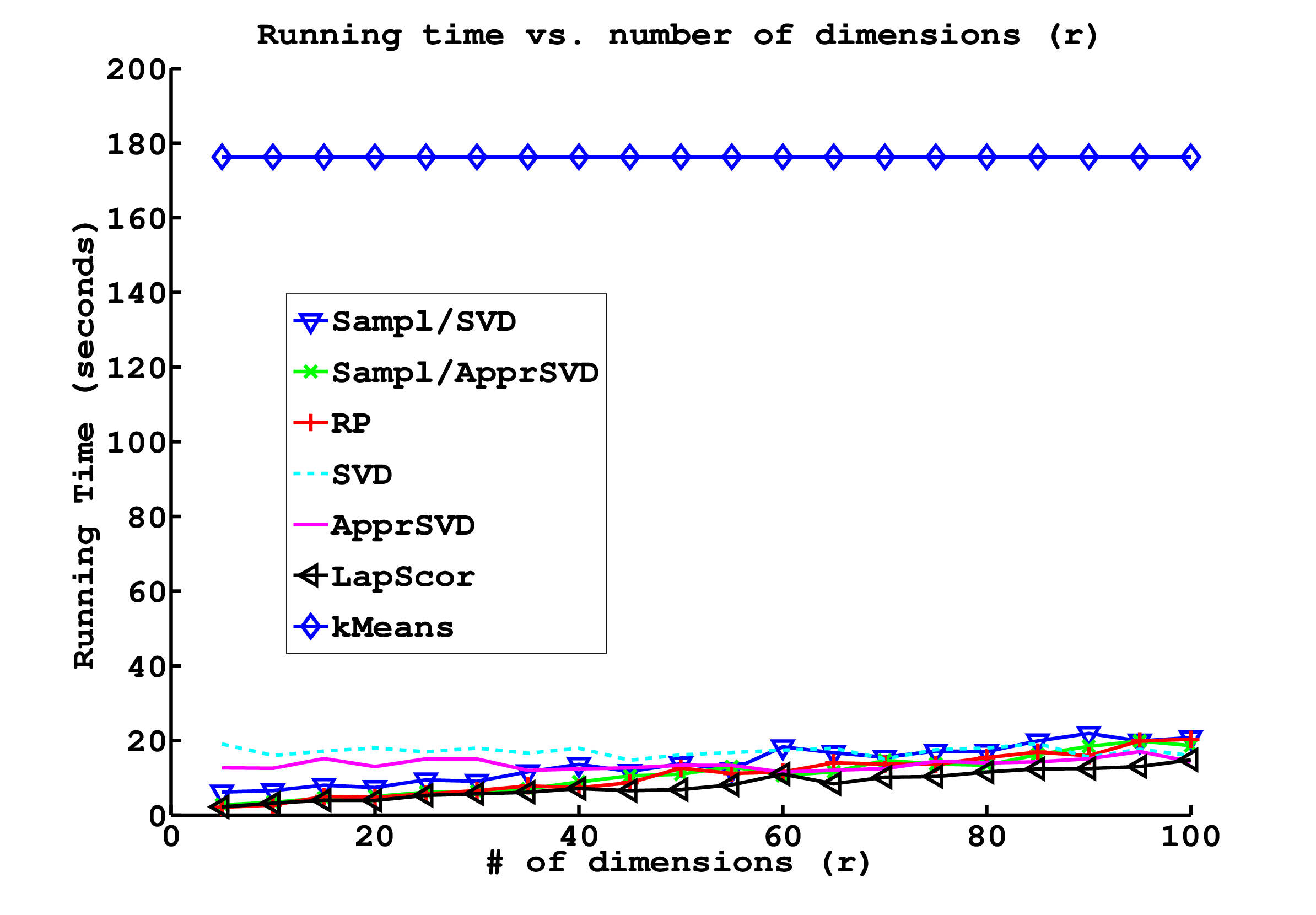}}
\subfigure[\texttt{ORL} - Running time]{\includegraphics[width=0.45\textwidth]{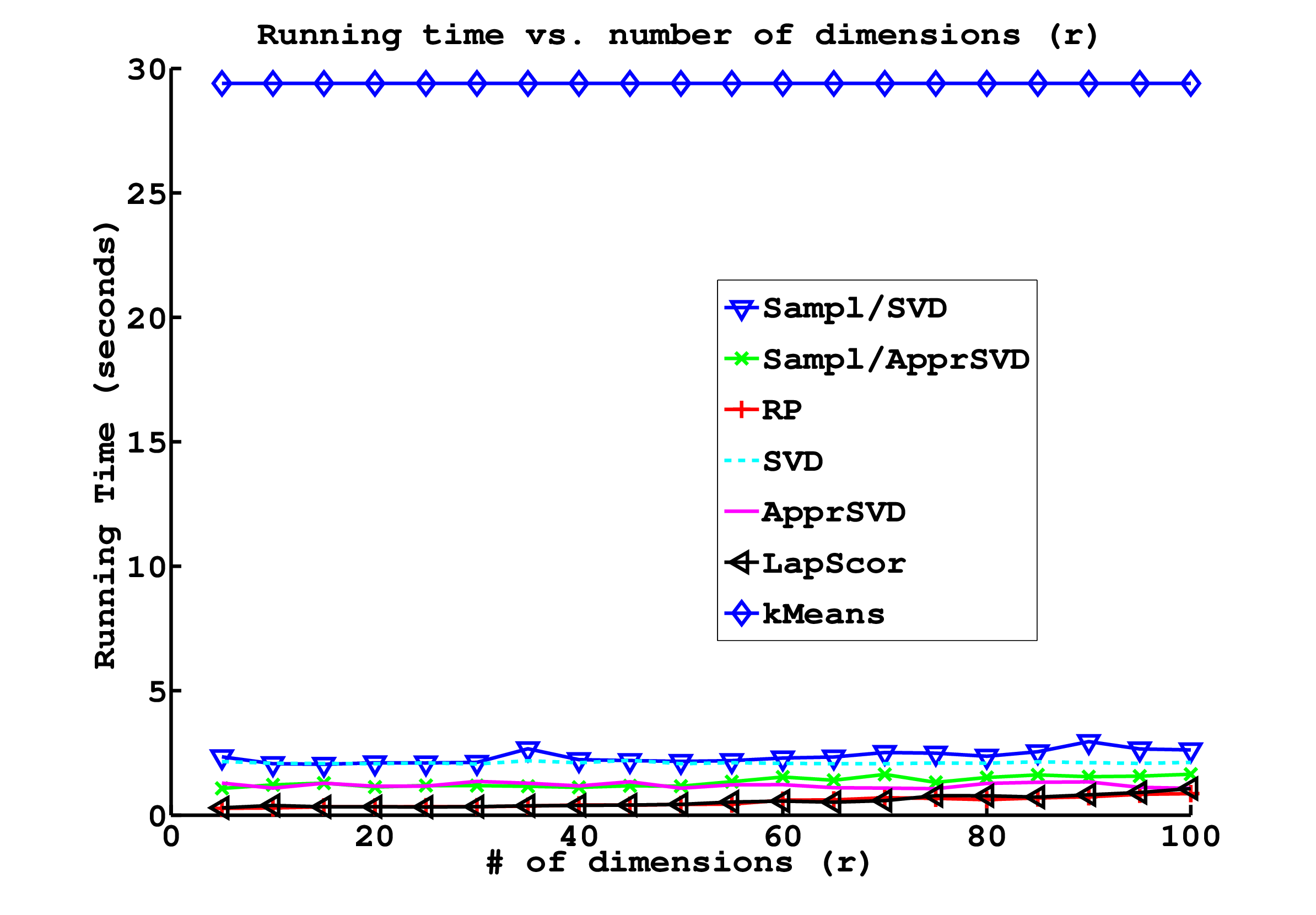}}\\
\subfigure[\texttt{PIE} - Objective]{\includegraphics[width=0.45\textwidth]{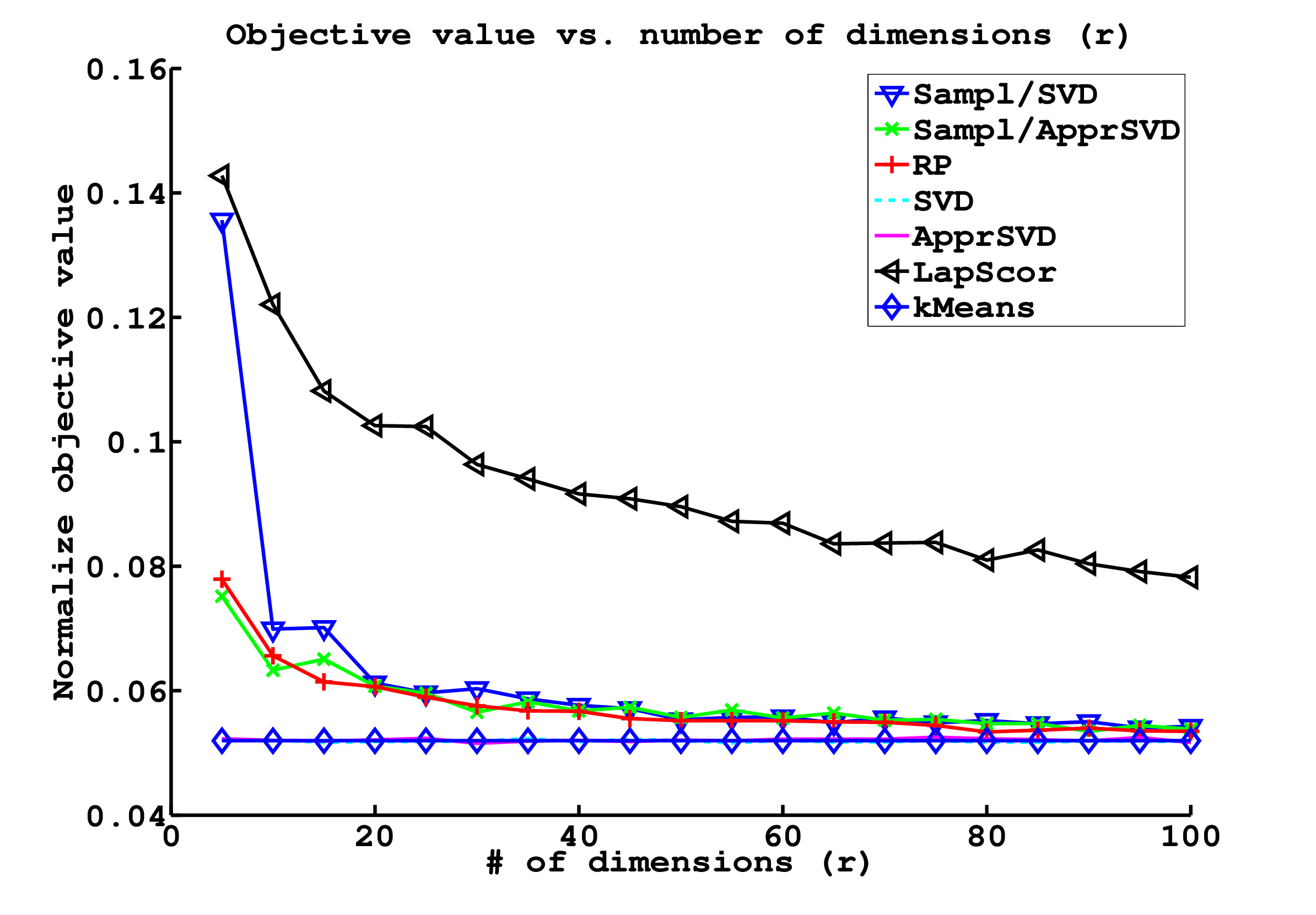}}
\subfigure[\texttt{ORL} - Objective]{\includegraphics[width=0.45\textwidth]{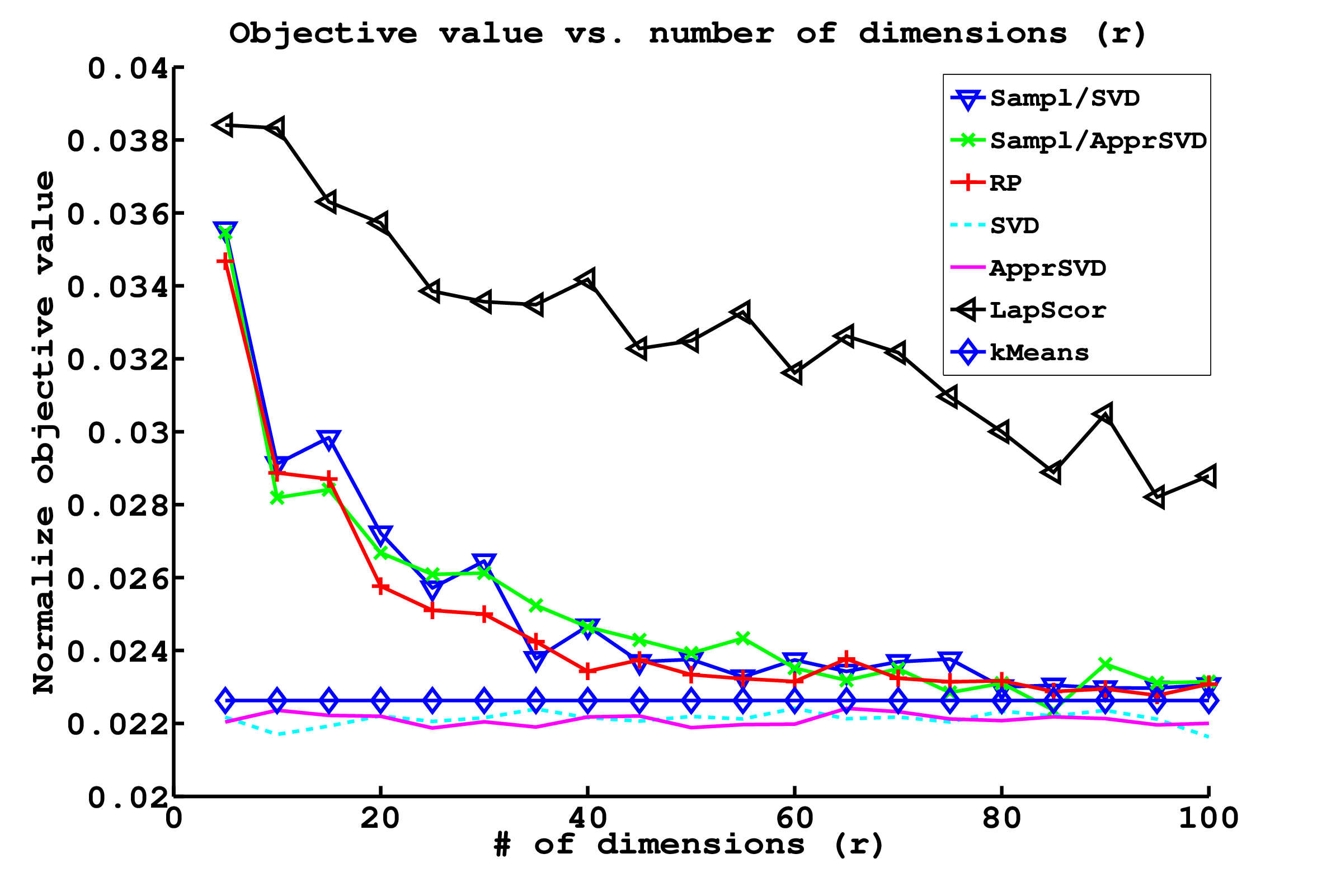}}\\
\subfigure[\texttt{PIE} - Accuracy of clustering]{\includegraphics[width=0.45\textwidth]{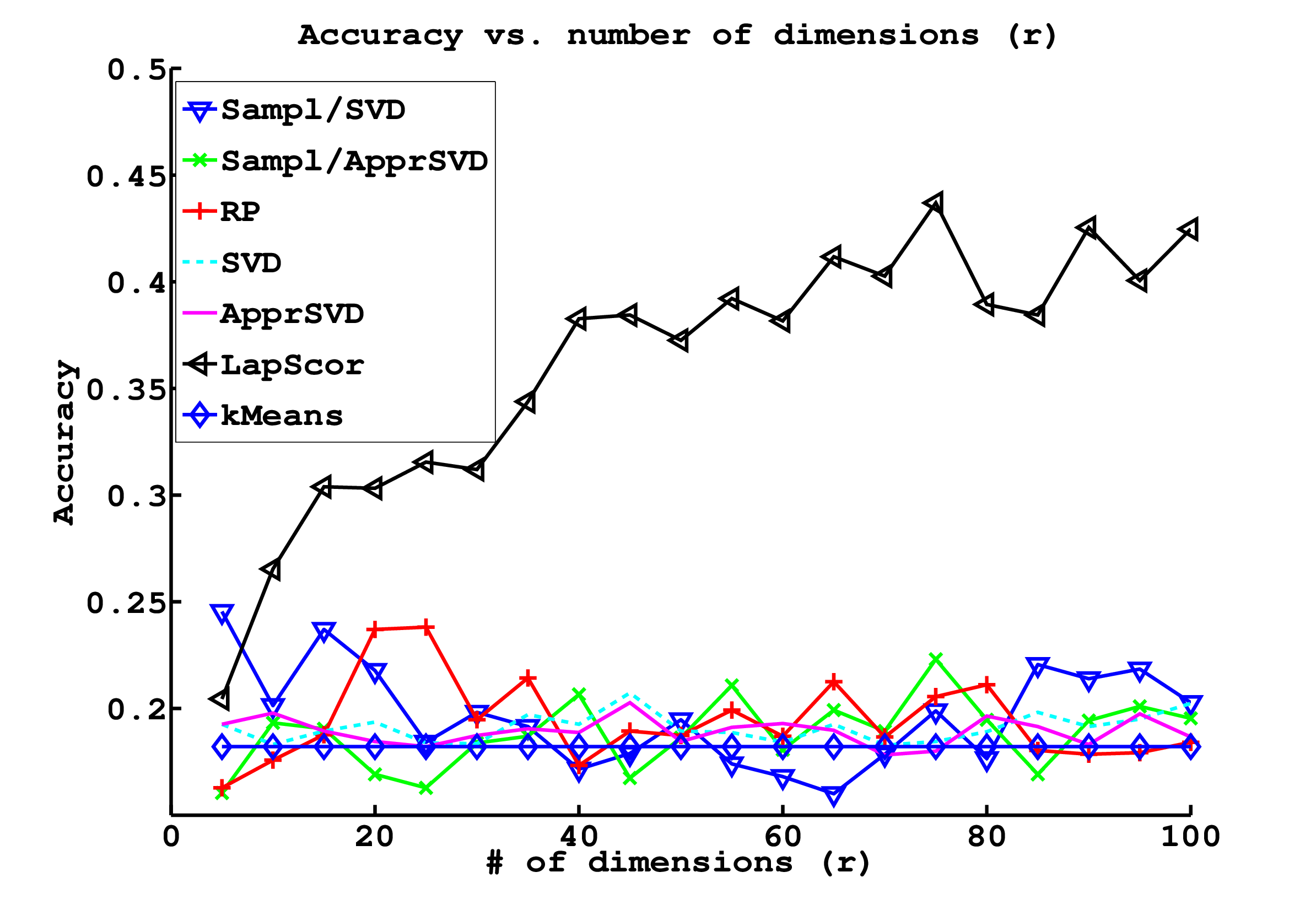}}
\subfigure[\texttt{ORL} - Accuracy of clustering]{\includegraphics[width=0.45\textwidth]{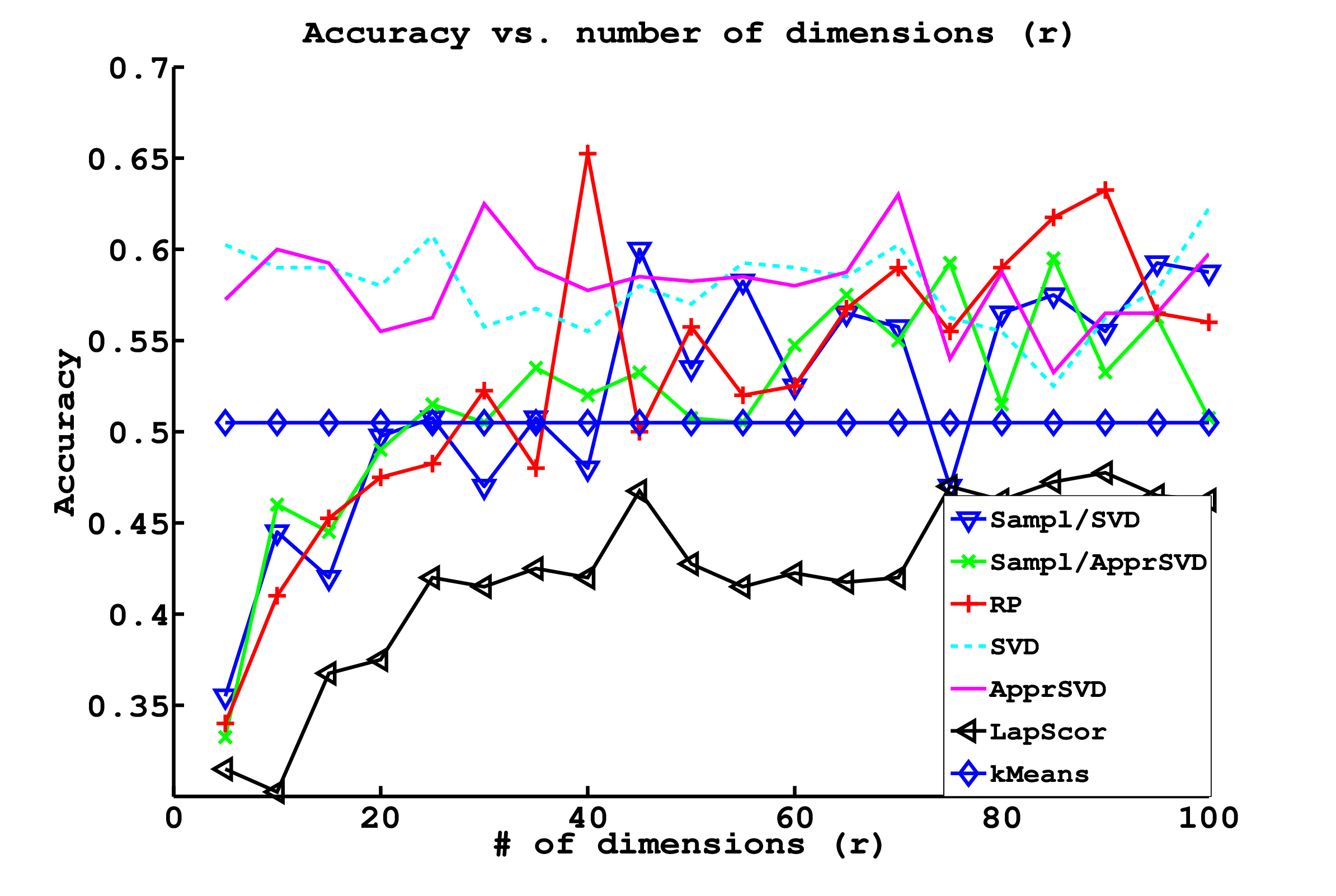}}\\
\caption{\small Plot of running time  (a),(b), objective value (c),(d) and accuracy versus (e),(f) the number of projected dimensions for several dimensionality reduction approaches. Left column corresponds to the \texttt{PIE} dataset, whereas the right column corresponds to the \texttt{ORL} dataset.}
\label{fig:orl}
\end{figure}
\subsection{Results}
We present the results of our experiments in Figures~\ref{fig:usps},\ref{fig:light} and \ref{fig:orl}. We experimented with relative small values for the number of dimensions: 
\[r = 5, 10, 15, \dots, 100.\]
In the synthetic dataset, we observe that all dimensionality reduction methods for $k$-means clustering are clearly more efficient compared to naive $k$-means clustering. More importantly, the accuracy plots of Figure~\ref{fig:usps} demonstrate that the dimensionality reduction approach is also accurate in this case even for relatively (with respect to $k$) small values of $r$, i.e., $\approx 20$. Recall that in this case the clusters of the dataset are well-separated between each other. Hence, these observations suggest that dimensionality reduction for $k$-means clustering is effective when applied to well-separated data points. 

The behavior of the dimensionality reduction methods for $k$-means clustering for the real-world datasets is similar with the synthetic dataset, see Figures~\ref{fig:light} and~\ref{fig:orl}. That is, as the number of projecting dimensions increases, the normalized objective value of the resulting clustering decreases. Moreover, in all cases the normalized objective value of the proposed methods converge to the objective value attained by the naive $k$-means algorithm (as the number of dimensions increases). In all cases but the \texttt{PIE} and \texttt{COIL20} dataset, the proposed dimensionality reduction methods have superior performance compared to Laplacian Scores~\cite{HCN06} both in terms of accuracy and normalized $k$-means objective value. In the \texttt{PIE} and \texttt{COIL20} datasets, the Laplacian Scores approach is superior compared to all other approaches in terms of accuracy. However, notice that in these two datasets the naive $k$-means algorimhm performs poorly in terms of accuracy which indicates that the data might not be well-separated.

Regarding the running times of the algorithms notice that in some cases the running time does not necessarily increased by increasing the number of dimensions. This happens because after the dimensionality reduction step the $k$-means method might take a different number of iterations to converge. We did not investigated this behavior further since this is not the focus of our experimental evaluation. 

Our experiments indicate that running our algorithms with small values of $r$, e.g., $r=20$ or $r=30$, achieves nearly optimal separation of a mixture of Gaussians and does well in several real-world clustering problems. Although a more thorough experimental evaluation of our algorithms would have been far more informative, our preliminary experimental findings are quite encouraging with respect to the performance of our algorithms in \emph{practice}. 

\section{Conclusions}

We studied the problem of dimensionality reduction for $k$-means clustering. Most of the existing results in this topic consist of heuristic approaches, whose excellent empirical performance can not be explained with a rigorous theoretical analysis. In this paper, our focus was on dimensionality reduction methods that work well \emph{in theory}. We presented three such approaches, one feature selection method for $k$-means and two feature extraction methods. The theoretical analysis of the proposed methods is based on the fact that dimensionality reduction for $k$-means has deep connections with low-rank approximations to the data matrix that contains the points one wants to cluster. We explained those connections in the text and employed modern fast algorithms to compute such low rank approximations and designed fast algorithms for dimensionality reduction in $k$-means.

Despite our focus on the theoretical foundations of the proposed algorithms, we tested the proposed methods in practice and concluded that the experimental results are very encouraging: dimensionality reduction for $k$-means using the proposed techniques leads to faster algorithms that are almost as accurate as running $k$-means on the high dimensional data. 

All in all, our work describes the \emph{first} provably efficient feature selection algorithm for $k$-means clustering as well as two novel provably efficient feature extraction algorithms. 
An interesting path for future research is to design provably efficient $(1+\varepsilon)$-relative error dimensionality reduction methods for $k$-means. 


\bibliographystyle{abbrv}
\bibliography{bibfile}

\clearpage

\appendix
\section*{Appendix}
\label{app:theorem}




\section*{}

%
%

\noindent The following technical lemma is useful in the proof of Lemma~\ref{lem:rsall} and the proof of Lemma~\ref{lem:rpall}.
\begin{lemma}\label{lem:pseudoinv_transp}
	Let $\matQ\in\R^{n\times k}$ with $n > k$ and $\matQ\transp \matQ = \matI_k$. Let $\matTheta$ be any $n\times r$ matrix ($r > k$) satisfying $ 1-\varepsilon \leq \sigma^2_i (\matQ\transp \matTheta) \leq 1+\varepsilon$ for every $i=1,\dots , k$ and $0< \varepsilon < 1/3$. Then, $$\TNorm{ (\matQ\transp \matTheta)^\dagger - (\matQ\transp \matTheta)\transp } \leq \frac{  \varepsilon }{\sqrt{1-\varepsilon}}  \le 1.5 \varepsilon.$$
\end{lemma}
%
\begin{proof}
Let $\matX = \matQ\transp \matTheta \in \R^{k \times r}$
with SVD $\matX = \matU_{\matX} \matSig_{\matX} \matV_\matX\transp $.
Here,  $\matU_{\matX} \in \R^{k \times k}$,
$\matSig_{\matX} \in \R^{k \times k}$, and $\matV_\matX \in \R^{r \times k}$, since $r > k$.
Consider taking the SVD of $(\matQ\transp \matTheta)^\dagger$ and $(\matQ\transp \matTheta)\transp $,
\[ \TNorm{(\matQ\transp \matTheta)^\dagger - (\matQ\transp \matTheta)\transp} =
    \TNorm{\matV_{\matX} \matSig_{\matX}^{-1} \matU_{\matX}\transp - \matV_{\matX} \matSig_{\matX} \matU_{\matX}\transp  } =
    \TNorm{\matV_{\matX}(\matSig_{\matX}^{-1} - \matSig_{\matX}) \matU_{\matX}\transp} = \TNorm{\matSig_{\matX}^{-1} - \matSig_{\matX}},\]
since $\matV_{\matX}$ and $\matU_{\matX}\transp $ can be dropped without changing the spectral norm.
Let $\matY = \matSig_{\matX}^{-1} - \matSig_{\matX} \in \R^{k \times k}$ be
a diagonal matrix. Then, for all
$i=1,\ldots ,k$,
$\matY_{ii}  = \frac{ 1 - \sigma_i^2(\matX)  }{ \sigma_{i}(\matX) }.$
Since $\matY$ is diagonal,
$$
\TNorm{ \matY } =
\max_{1 \leq i \leq k} \abs{\matY_{ii} } =
\max_{1 \leq i \leq k} \abs{\frac{ 1 - \sigma_i^2(\matX)}{ \sigma_{i}(\matX) } } =
\max_{1 \leq i \leq k} \frac{ \abs{1 - \sigma_i^2(\matX)} }{ \sigma_{i}(\matX) } \le
\frac{  \varepsilon }{\sqrt{1-\varepsilon}} \le 1.5 \varepsilon.
$$
The first equality follows since the singular values are positive (from our choice of $\varepsilon$ and the left hand side of the
bound for the singular values). The first inequality follows by the bound for the singular values of $\matX$. The last inequality follows by the
assumption that $0 < \varepsilon < 1/3$.
\end{proof}

\subsection*{Proof of Lemma \ref{lem:rsall}}
\begin{proof}
We begin with the analysis of a matrix-multiplication-type term involving the multiplication
of the matrices $\matE$, $\matZ$. The sampling and rescaling matrices $\matOmega, \matS$
indicate the subsampling of the columns and rows of $\matE$, $\matZ$, respectively.
Eqn.~(4) of Lemma 4 of \cite{DKM06a} gives a bound for such $\matOmega, \matS$
constructed with randomized sampling with replacement and \emph{any} set of probabilities
$p_1, p_2, \dots ,p_n$ (over the columns of $\matE$ - rows of $\matZ$),
$$ \Expect{ \FNormS{ \matE \matZ - \matE \matOmega \matS \matS\transp \matOmega\transp \matZ } } \le
\sum_{i=1}^{n} \frac{ \TNormS{\matE^{(i)}} \TNormS{ \matZ_{(i)} } }{r p_i} - \frac{1}{r} \FNormS{\matE \matZ}. $$
Notice that $\matE \matZ = \bm{0}_{m \times k}$, by construction (see Lemma~\ref{tropp2}).
Now, for every $i=1,\dots , n$ replace the values $p_i = \frac{ \TNormS{ \matZ_{(i)}  }}{k}$
(in Definition~\ref{def:sampling}) and rearrange,
\begin{equation}\label{ineq:RS:matrix_mult}
\Expect{ \FNormS{\matE \matOmega \matS \matS\transp \matOmega\transp \matZ } } \le \frac{k}{r} \FNormS{\matE}.
\end{equation}
\noindent Observe that Lemma~\ref{lem:random} and our choice of $r$, implies that w.p. $1 - \delta$,
\begin{equation}\label{ineq:RS:preserve_sigma}
	1 -  \varepsilon  \leq  \sigma_i^2(\matZ\transp \matOmega \matS)   \leq 1 + \varepsilon,\quad \text{for all } i =1,\ldots ,k.
\end{equation}
For what follows, condition on the event of Ineq.~\eqref{ineq:RS:preserve_sigma}. First, $\sigma_k(\matZ\transp \matOmega \matS) > 0$. So, $\rank(\matZ\transp \matOmega \matS) = k$ and $ (\matZ\transp \matOmega \matS)(\matZ\transp \matOmega \matS)^\dagger = \matI_{k}$\footnote[5]{To see this, let $\matB = \matZ\transp \matOmega \matS \in \R^{k \times r}$ with
SVD $\matB = \matU_{\matB} \matSig_{\matB} \matV_\matB\transp $.
Here,  $\matU_{\matB} \in \R^{k \times k}$,
$\matSig_{\matB} \in \R^{k \times k}$, and $\matV_\matB \in \R^{r \times k}$, since $r > k$.
Finally, $(\matZ\transp \matOmega \matS)(\matZ\transp \matOmega \matS)^\dagger
= \matU_{\matB} \matSig_{\matB} \underbrace{\matV_\matB\transp \matV_\matB}_{\matI_k} \matSig^{-1}_{\matB} \matU_{\matB}\transp
= \matU_{\matB} \underbrace{\matSig_{\matB} \matSig^{-1}_{\matB}}_{\matI_k} \matU_{\matB}\transp = \matI_{k}$.}.
Now,
$ \matA\matZ \matZ\transp - \matA\matZ \matZ\transp \matOmega \matS (\matZ\transp \matOmega \matS)^\dagger\matZ\transp =
\matA\matZ \matZ\transp - \matA\matZ \matI_{k} \matZ\transp = \bm{0}_{m \times n}.$
Next, we manipulate the term
$\theta = \FNorm{ \matA\matZ \matZ\transp - \matA \matOmega \matS (\matZ\transp \matOmega \matS)^\dagger\matZ\transp}$ as follows
(recall, $\matA = \matA\matZ \matZ\transp + \matE$),
\begin{eqnarray*}
\theta =
\FNorm{
\underbrace{ \matA\matZ \matZ\transp - \matA\matZ \matZ\transp \matOmega \matS (\matZ\transp \matOmega \matS)^\dagger\matZ\transp}_{\bm{0}_{m \times n}}
- \matE \matOmega \matS (\matZ\transp \matOmega \matS)^\dagger\matZ\transp }
         = \FNorm{\matE \matOmega \matS (\matZ\transp \matOmega \matS)^\dagger\matZ\transp }.
\end{eqnarray*}
Finally, we manipulate the latter term as follows,
\begin{eqnarray*}
\FNorm{\matE \matOmega \matS (\matZ\transp \matOmega \matS)^\dagger\matZ\transp }
&\le& \FNorm{\matE \matOmega \matS (\matZ\transp \matOmega \matS)^\dagger } \\
&\le& \FNorm{\matE \matOmega \matS (\matZ\transp \matOmega \matS)\transp} + \FNorm{\matE \matOmega \matS}
\TNorm{(\matZ\transp \matOmega \matS)^\dagger - (\matZ\transp \matOmega \matS)\transp} \\
&\le& \sqrt{\frac{k}{\delta r}} \FNorm{\matE} + \frac{1}{\sqrt{\delta}} \FNorm{\matE} \frac{\varepsilon}{\sqrt{1-\varepsilon}}
\le \left( \sqrt{ \frac{k}{\delta r} } + \frac{\varepsilon}{ \sqrt{\delta} \sqrt{1-\varepsilon}} \right) \FNorm{\matE}  \\
&\le& \left( \frac{\varepsilon}{2 \sqrt{\delta}} \frac{1}{\sqrt{\ln(2k/\delta)}} + \frac{\varepsilon}{\sqrt{\delta}\sqrt{1-\varepsilon}} \right)\FNorm{\matE} \\
&\le& \left( \frac{\varepsilon}{2 \ln(4) \sqrt{\delta}}  + \frac{\varepsilon}{\sqrt{\delta}\sqrt{1-\varepsilon}} \right)\FNorm{\matE}
\le \frac{ 1.6 \varepsilon}{\sqrt{\delta}}  \FNorm{\matE}.
\end{eqnarray*}
The first inequality follows by spectral submultiplicativity and the fact that $\TNorm{\matZ\transp}=1$.
The second inequality follows by the triangle inequality for matrix norms.
In the third inequality, the bound for the term $\FNorm{\matE \matOmega \matS (\matZ\transp \matOmega \matS)\transp}$ follows by applying to it Markov's inequality together with Ineq.~\eqref{ineq:RS:matrix_mult}; also, $\FNorm{\matE \matOmega \matS}$ is bounded by $ (1/\sqrt{\delta}) \FNorm{\matE}$ w.p. $1-\delta$ (Lemma~\ref{lem:fnorm}), while we bound $\TNorm{(\matZ\transp \matOmega \matS)^\dagger - (\matZ\transp \matOmega \matS)\transp}$ using Lemma~\ref{lem:pseudoinv_transp} (set $\matQ = \matZ$ and $\matTheta = \matOmega \matS$ 
). So, by the union bound, the failure probability is $3\delta$.
The rest of the argument follows by our choice of $r$, assuming $k \ge 2$, $\varepsilon < 1/3$ and simple algebraic manipulations.
\end{proof}
\subsection*{Proof of Lemma~\ref{lem:RP:FNorm}}
\begin{proof}
First, define the random variable $Y = \FNorm{\matY \matR }^2$. It is easy to see that $\Expect{Y} = \FNorm{\matY}^2$ and moreover an upper bound for the variance of $Y$ is available in Lemma~$8$ of \cite{Sar06}: $\var{Y} \leq 2 \FNorm{\matY}^4 /r$ \footnote[6]{\cite{Sar06} assumes that the matrix $\matR$ has i.i.d rows, each one containing
four-wise independent zero-mean $\{1/\sqrt{r}, -1/\sqrt{r}\}$ entries. The claim in our lemma follows because our rescaled sign matrix $\matR$ satisfies the four-wise independence assumption, by construction.
}. Now, Chebyshev's inequality tells us that,
\begin{eqnarray*}
\Prob ( |Y - \Expect{Y }| \geq \varepsilon \FNorm{\matY}^2 )
\le \frac{\var{Y}}{\varepsilon^2 \FNorm{\matY}^4}
\le \frac{2\FNorm{\matY}^4}{r\varepsilon^2 \FNorm{\matY}^4}
\le \frac{2}{c_0 k}
\le 0.01.
\end{eqnarray*}
The last inequality follows by assuming $c_0\geq 100$ and the fact that $k > 1$.
Finally, taking square root on both sides concludes the proof.
\end{proof}
%
%
%
\begin{proof}[Proof of Lemma \ref{lem:rpall}]
%
\noindent
We start with the definition of the Johnson-Lindenstrauss transform.
\begin{definition}[Johnson-Lindenstrauss Transform]
\label{lem:jlt}
A random matrix $\matR \in \R^{n\times r}$ forms a \emph{Johnson-Lindenstrauss transform} if, for any (row) vector $\vct{x}\in\R^n$,
\[  \Prob{\left( \left(1-\varepsilon\right) \TNormS{ \vct{x} } \leq \TNormS{ \vct{x}\matR } \leq \left(1+\varepsilon\right)\TNormS{ \vct{x}} \right)} \geq 1 - e^{-C \varepsilon^2 r}\]
where $C>0$ is an absolute constant.
\end{definition}
Notice that in order to achieve failure probability at most $\delta$, it suffices to take $r= O( \log (1/\delta) /\varepsilon^2) $. We continue with Theorem $1.1$ of \cite{Ach03} (properly stated to fit our notation
and after minor algebraic manipulations), which indicates that a (rescaled) sign matrix $\matR$ corresponds to a \emph{Johnson-Lindenstrauss transform} as defined above.
\begin{theorem}[\cite{Ach03}]\footnote[7]{This theorem is proved by first showing that a rescaled random sign matrix is a \emph{Johnson-Lindenstrauss transform}~\cite[Lemma~$5.1$]{Ach03} with constant $C=36$. Then, setting an appropriate value for $r$ and applying the union bound over all pairs of row indices of $\matA$ concludes the proof.}
Let $\matA \in \R^{m \times n}$ and $0 < \varepsilon < 1$. Let $\matR \in \R^{n \times r}$ be a rescaled random sign matrix
with  $r = \frac{36}{\varepsilon^2} \log (m) \log(1/\delta)$.
Then for all $i,j=1,\ldots, m$ and w.p. at least $1 - \delta$,
\[
(1- \varepsilon) \TNormS{ \matA_{(i)} - \matA_{(j)} }
\leq  \TNormS{ \left(\matA_{(i)} - \matA_{(j)} \right)\matR } \leq
(1 + \varepsilon) \TNormS{ \matA_{(i)} - \matA_{(j)}}.
\]
\end{theorem}
In addition, we will use a matrix multiplication bound which follows from Lemma~$6$ of~\cite{Sar06}. The second claim
of this lemma says that for any $\matX \in \R^{m \times n}$ and $\matY \in \R^{n \times p}$,
if $\matR \in \R^{n \times r}$ is a matrix with i.i.d rows, each one containing
four-wise independent zero-mean $\{1/\sqrt{r}, -1/\sqrt{r}\}$ entries, then,
\begin{equation}\label{ineq:random_proj:MM}
	 \Expect{ \FNormS{ \matX \matY - \matX \matR \matR\transp \matY } } \le \frac{2}{r} \FNormS{ \matX } \FNormS{ \matY }.
\end{equation}
Our random matrix $\matR$ uses full independence, hence the above bound holds by dropping the limited independence condition.
%
%
%
%
%
\paragraph{Statement $1$.}
%
%
The first statement in our lemma has been proved in Corollary $11$ of~\cite{Sar06}, see also~\cite[Theorem~$1.3$]{Clark08} for a restatement. More precisely, repeat the proof of Corollary $11$ of~\cite{Sar06} paying attention to the constants. That is, set $\matC=\matV_k\transp \matR\transp \matR \matV_k -\matI_k$ and $\varepsilon_0=1/2$ in Lemma~$10$ of~\cite{Sar06}, and apply our JL transform with (rescaled) accuracy $\varepsilon/4$ on each vector of the set $T':=\{\matV_k\transp \vct{x}\ |\ \vct{x} \in T \}$ (which is of size at most $\leq e^{k\ln (18)}$, see \cite[Lemma~4]{random06} for this bound). So,
\begin{equation}\label{ineq:sigma_preserved}
	\Prob{\left(  \forall i = 1,\dots ,k :\  1-\varepsilon \leq \sigma_i^2(\matV_k\transp \matR) \leq 1+\varepsilon \right) } \geq 1 - e^{k\ln (18) } e^{-\varepsilon^2 r/(36 \cdot 16) } .
\end{equation}
Setting $r$ such that the failure probability is at
most $0.01$ indicates that $r$ should be at least $ r \ge 576 ( k \ln(18) + \ln(100) )/\varepsilon^2 $. So,
$c_0 = 3330$ is a sufficiently large constant for the lemma.
%
%
%
%
\paragraph{Statement $2$.}
%
Consider the following three events (w.r.t. the randomness of the random matrix $\matR$): $\mathcal{E}_1:=\{ 1-\varepsilon \leq \sigma_i^2(\matV_k\transp \matR) \leq 1+\varepsilon\}$,
$\mathcal{E}_2:=\{ \FNorm{\matA_{\rho-k}\matR }^2 \leq (1+\varepsilon) \FNorm{\matA_{\rho-k} }^2 \}$ and $\mathcal{E}_3:=\{ \FNorm{\matA_{\rho-k}\matR \matR\transp \matV_{k}}^2 \leq \varepsilon^2 \FNorm{\matA_{\rho-k}}^2 \}$. Ineq.~\eqref{ineq:sigma_preserved} and Lemma~\ref{lem:RP:FNorm} with $\matY = \matA_{\rho-k}$ imply that $\Prob{ (\mathcal{E}_1)} \geq 0.99$, $\Prob{(\mathcal{E}_2)} \geq 0.99$, respectively. A crucial observation for bounding the failure probability of the last event $\mathcal{E}_3$ is that $\matA_{\rho-k}\matV_k=\matU_{\rho-k}\matSig_{\rho-k}\matV_{\rho-k}^{\top}\matV_k=\mathbf{0}_{m \times k}$ by orthogonality of the columns of $\matV_k$ and $\matV_{\rho-k}$. This event can now be bounded by applying Markov's Inequality on Ineq.~\eqref{ineq:random_proj:MM} with $\matX = \matA_{\rho - k}$ and $\matY = \matV_k$ and recalling that $\FNorm{\matV_k}^2 =k$ and $r=c_0k/\varepsilon^2$. Assuming $c_0\geq 200$, it follows that $\Prob{(\mathcal{E}_3)}\geq 0.99$ (hence, setting $c_0=3330$ is a sufficiently large constant for both statements). A union bound implies that these three events happen w.p. $0.97$. For what follows, condition on these three events.

Let $\widetilde{\matE} = \matA_k  - (\matA \matR) (\matV_k\transp  \matR)^\dagger \matV_k\transp  \in \R^{m \times n}$.
By setting $\matA = \matA_k + \matA_{\rho-k}$ and using the triangle inequality,
\[ \FNorm{\widetilde{\matE}}\ \leq\ \FNorm{\matA_k - \matA_k \matR (\matV_k\transp \matR)^\dagger \matV_k\transp }\ +\ \FNorm{ \matA_{\rho-k}\matR (\matV_k\transp \matR)^\dagger \matV_k\transp  }.\]
The event $\mathcal{E}_1$ implies that
$\text{rank}(\matV_k\transp \matR) = k$ thus\footnote[8]{To see this, let $\matB = \matV_k\transp \matR \in \R^{k \times r}$ with
SVD $\matB = \matU_{\matB} \matSig_{\matB} \matV_\matB\transp $.
Here,  $\matU_{\matB} \in \R^{k \times k}$,
$\matSig_{\matB} \in \R^{k \times k}$, and $\matV_\matB \in \R^{r \times k}$, since $r > k$.
Finally, $(\matV_k\transp \matR)(\matV_k\transp \matR)^\dagger
= \matU_{\matB} \matSig_{\matB} \underbrace{\matV_\matB\transp \matV_\matB}_{\matI_k} \matSig^{-1}_{\matB} \matU_{\matB}\transp
= \matU_{\matB} \underbrace{\matSig_{\matB} \matSig^{-1}_{\matB}}_{\matI_k}\matU_{\matB}\transp = \matI_{k}$.},
$$(\matV_k\transp \matR) (\matV_k\transp \matR)^\dagger =
\matI_k.$$
Replacing
$\matA_k = \matU_k \matSig_k \matV_k\transp $ and setting
$(\matV_k\transp \matR) (\matV_k\transp \matR)^\dagger = \matI_k,$
we obtain that
\[
\FNorm{\matA_k - \matA_k \matR (\matV_k\transp \matR)^\dagger \matV_k\transp }
=
\FNorm{\matA_k - \matU_k\matSig_k \underbrace{\matV_k\transp \matR (\matV_k\transp \matR)^\dagger}_{\matI_k} \matV_k\transp } =
\FNorm{\matA_k - \matU_k \matSig_k \matV_k\transp } = 0. \]
To bound the second term above, we drop $\matV_k\transp $, add and
subtract $\matA_{\rho-k}\matR (\matV_k\transp \matR)\transp   $,
and use the triangle inequality and spectral sub-multiplicativity,
\begin{eqnarray*}
\FNorm{ \matA_{\rho-k} \matR (\matV_k\transp \matR)^\dagger \matV_k\transp  } & \leq & \FNorm{ \matA_{\rho-k}\matR (\matV_k\transp \matR)\transp  }\ +\ \FNorm{ \matA_{\rho-k}\matR ( (\matV_k\transp \matR)^\dagger - (\matV_k\transp \matR)\transp )} \\
                                               & \leq & \FNorm{ \matA_{\rho-k}\matR \matR\transp  \matV_k  }\ +\ \FNorm{ \matA_{\rho-k}\matR} \TNorm{ (\matV_k\transp \matR)^\dagger - (\matV_k\transp \matR)\transp                                           }.
\end{eqnarray*}
Now, we will bound each term individually.  We bound the first term using $\mathcal{E}_3$. The second term can be bounded using $\mathcal{E}_1$ and $\mathcal{E}_2$ together with Lemma~\ref{lem:pseudoinv_transp} (set $\matQ =\matV_k$ and $\matTheta = \matR$).  Hence,
\begin{eqnarray*}
\FNorm{ \widetilde{\matE} }  & \leq & \FNorm{ \matA_{\rho-k}\matR \matR\transp  \matV_k  } + \FNorm{ \matA_{\rho-k}\matR}
\TNorm{ (\matV_k\transp \matR)^\dagger - (\matV_k\transp \matR)\transp } \\
                           & \leq &  \varepsilon \FNorm{ \matA_{\rho-k}} + \sqrt{(1+\varepsilon)} \FNorm{\matA_{\rho-k}} \cdot 1.5 \varepsilon \\
                           & \leq &  \varepsilon \FNorm{ \matA_{\rho-k}} +  2 \varepsilon \FNorm{\matA_{\rho-k}} \\
                           &  = &  3 \varepsilon\cdot \FNorm{\matA_{\rho-k}}.
\end{eqnarray*}
The last inequality holds by our choice of $\varepsilon \in (0,1/3)$.
\end{proof}

\subsection*{Proof of Eqn.~(\ref{proof:short})}
\begin{proof}
$\Expect{\FNormS{\matE}} \le (1+\varepsilon) \FNormS{\matA -\matA_k} \rightarrow
\Expect{\FNormS{\matE} - \FNormS{\matA -\matA_k} \le \varepsilon \FNormS{\matA -\matA_k} }.
$ Now, apply Markov's inequality on the random variable $Y = \FNormS{\matE} - \FNormS{\matA -\matA_k} \ge 0$.
($Y \ge 0$ because $\matE = \matA - \matA\matZ\matZ\transp$ and $\rank(\matA\matZ\matZ\transp) = k$). This gives
$\FNormS{\matE} - \FNormS{\matA -\matA_k} \le 100 \varepsilon \FNormS{\matA -\matA_k}$ w.p. $0.99$; so,
$\FNormS{\matE}  \le \FNormS{\matA -\matA_k} + 100 \varepsilon \FNormS{\matA -\matA_k}$.
\end{proof}

\end{document}